\documentclass[%
 reprint,
nofootinbib,
 amsmath,amssymb,
 aps,
 longbibliography,
]{revtex4-2}

\usepackage{graphicx}% Include figure files
\usepackage{dcolumn}% Align table columns on decimal point
\usepackage{bm,mathtools,bbm, amsthm, comment}% bold math
\usepackage{newfloat}
\DeclareFloatingEnvironment[
    fileext=los,
    listname={List of Schemes},
    name=Protocol,
    placement=tbhp,
]{protocol}
\usepackage{mdframed}
\mdfdefinestyle{mystyle}{
}

\newtheorem{definition}{Definition}[section]
\newtheorem{theorem}{Theorem}[section]

\newtheorem{result}{Result}[section]

\DeclareMathOperator{\rlangle}{\!\rangle \langle\!}

\DeclareMathOperator{\rrangle}{\!\rangle\!\rangle\!}
\DeclareMathOperator{\llangle}{\!\langle\!\langle\!}
\DeclareMathOperator{\rrllangle}{\!\rangle\!\rangle\!\langle\!\langle\!}
\DeclareMathOperator{\Tr}{\text{Tr}}
\DeclareMathOperator{\sums}{\sum \limits_{\mathbf{x} \in \mathbf{X}} \sum_{\mathfrak{B}\in\mathsf{C}} \sum_{\mathit{i,j}\in \mathfrak{B}}}

\usepackage{kbordermatrix}
\renewcommand{\kbldelim}{(}% Left delimiter
\renewcommand{\kbrdelim}{)}% Right delimiter

\usepackage{tikz}
\usepackage{mathtools}
\usetikzlibrary{arrows.meta, shapes.geometric, arrows}
\usetikzlibrary{quantikz2}
\tikzstyle{box} = [rectangle, rounded corners, minimum width=3cm, minimum height=1cm,text centered, draw=black, fill=orange!30]
\tikzstyle{box2} = [rectangle, rounded corners, minimum width=3cm, minimum height=1cm,text centered, draw=black, fill=green!20]
\tikzstyle{box3} = [rectangle, rounded corners, minimum width=3cm, minimum height=1cm,text centered, draw=black, fill=blue!20]
\tikzstyle{io} = [trapezium, trapezium left angle=70, trapezium right angle=110, minimum width=3cm, minimum height=1cm, text centered, draw=black, fill=blue!30]
\tikzstyle{process} = [rectangle, minimum width=3cm, minimum height=1cm, text centered, text width=3cm, draw=black, fill=orange!30]
\tikzstyle{decision} = [diamond, minimum width=3cm, minimum height=1cm, text centered, draw=black, fill=green!30]
\tikzstyle{arrow} = [thick,->,>=stealth]

\begin{document}

\preprint{APS/123-QED}

\title{Indefinite causal key distribution} 

\begin{comment}
\author{---}
\email{---}
\affiliation{---
} 
\end{comment}

\author{Hector Spencer-Wood}
\email{hector.spencer-wood@glasgow.ac.uk}
\affiliation{School of Physics and Astronomy, University of Glasgow, Glasgow G12 8QQ, Scotland
}

\date{\today}% It is always \today, today,

\begin{abstract}
We propose a quantum key distribution (QKD) protocol that is carried out in an indefinite causal order (ICO). In QKD, one considers a setup in which two parties, Alice and Bob, share a key with one another in such a way that they can detect whether an eavesdropper, Eve, has learnt anything about the key. To our knowledge, in all QKD protocols proposed until now, Eve is detected by publicly comparing a subset of Alice and Bob's key and checking for errors. We find that a consequence of our protocol is that it is possible to detect eavesdroppers \textit{without} publicly comparing any information about the key. Indeed, we prove that it is not possible for eavesdroppers, performing any individual attack, to extract useful information about the shared key without inducing a nonzero probability of being detected. We also prove the security of this protocol against a class of individual eavesdropping attacks. The role ICO plays in causing unusual phenomena in quantum technologies is an important question. By considering it we find a two-way QKD protocol that exhibits a similar private detection feature, albeit with some interesting differences. After noting some implications of these differences and discussing some of the practicalities of our protocol, we conclude that this work is best considered as a first step in applying quantum cryptographic ideas in an ICO.
\end{abstract}

\maketitle

% Introduction and Overview
\section{Introduction}

In our everyday, classical world, we are used to events occurring in a well defined order: $A$ happens before $B$ or vice versa. However, it has been suggested that if our universe is to be governed by a quantum theory of gravity, the dynamical nature of causality in general relativity must coexist with the indefiniteness of states in quantum mechanics \cite{Hardy}. So, although in our classical world $A$ comes before $B$ or vice versa, the quantum world allows for both of these options to occur in superposition, in an {\it indefinite causal order} (ICO). Remarkably, without even venturing into the quantum gravitational regime, a physically realisable device, called the {\it quantum switch}, has been proposed that exhibits such an indefinite causal order \cite{Quantum_Switch_original, Caslav_Orig_Process, QuantumSwitch_Romero, Rubino_Experimental_witness, InstrumentsInICO}. The quantum switch has been applied in many different settings both theoretically and experimentally \cite{RozemaReview}. Doing so has hinted at advantages, or interesting differences, when compared with the corresponding definite causal situations, whether it be in computation \cite{Quantum_Switch_original, Comp_Advantage}, channel discrimination \cite{Channel_discrim}, communications \cite{Comm_Compl, Channel_ICO1,Channel_ICO2, Channel_ICO, Indef_Noise}, metrology \cite{Indef_Metrology,Metrology_experiment,Aaron_ICO,freyMetrology} or thermodynamics \cite{Q_Fridge,Q_battery1,Q_battery2,  simonovWork}. Such differences usually present themselves when the actions of the parties within the quantum switch are incompatible (that is, when their operations do not commute with one another). This realisation has led to far reaching proposals: from quantum refrigeration \cite{Q_Fridge} and battery charging \cite{Q_battery1,Q_battery2}, to the violation of fundamental quantum metrological limits \cite{Indef_Metrology,Metrology_experiment}. Motivated by this relation to non-commuting operations, we explore whether another such application: quantum key distribution (QKD), can be performed in an ICO. Indeed, we ask: if it is possible, are there are any interesting consequences that follow?

%Cryptography has a long history made up of the battle between code makers and code breakers to, respectively, send and eavesdrop on secret messages \cite{CodeBook}. 

QKD is concerned with the scenario in which two parties, conventionally named Alice and Bob, would like to share a private key (a string of 0s and 1s) in such a way that they are confident an eavesdropping third party, called Eve, has not been listening in. There have been a number of such protocols proposed. Some are based on directly sending and measuring quantum states in one direction \cite{bb84, Decoy, SARG04} or two \cite{LM05, PingPong}, and others utilise quantum entanglement to derive their security \cite{E91, PingPong}. While the first QKD protocol was a one-way send and measure scheme proposed by Charles Bennett and Giles Brassard in 1984 (BB84) \cite{bb84}, cutting-edge protocols rely on entanglement to ensure practicality, and device-independent security \cite{colbeck_Device_independence, Fully_device_independent, RefRef2, RefRef3}. At their core, the security of these protocols comes from the fact that Eve can be detected. This is possible because, when Eve is present, the quantum phenomenon of measurement disturbance leads to a non-zero probability of error in Bob's key, with respect to Alice's. So, if one could somehow detect these errors induced by Eve, it could be concluded that an eavesdropper had been listening in. The way that these errors are normally detected is by having Alice and Bob publicly compare a subset of their respective raw keys. Now public information, this subset is subsequently discarded regardless of whether they conclude Eve is there or not. %In BB84, Alice prepares two-level quantum systems, called quantum bits or qubits, in a basis state of one of two non-orthogonal bases. She then sends them to Bob who measures them in one of the same two non-orthogonal bases. When no eavesdropper is present and Alice and Bob publicly confirm that Bob measured in the same basis that Alice prepared the state in, they are guaranteed to agree on the state of the system, assuming noiseless and lossless transmission. That is, they are guaranteed to agree on this key bit (single binary digit in the shared key). However, it turns out that,

To our knowledge, this public comparison is a feature of all QKD protocols so far proposed. In this article, we consider how one might adapt the simplest QKD scheme: the BB84 protocol, to an indefinite causal regime. In doing so, we find that we can determine whether eavesdroppers are there or not {\it without} having to publicly compare a subset of the distributed key. Indeed, we show that this is true for {\it any} individual attack performed by eavesdroppers, at least those who cannot infiltrate the causal structure of the protocol. We also provide some understanding of the security of our protocol by proving that it is secure against a class of individual attacks, until a detection probability of approximately $9.6\%$ is induced. If one assumes these are the only attacks available to eavesdroppers, this allows for a secure key to be obtained for error rates of up to $19.2\%$. It is natural to ask whether this ``private" detection is a consequence of ICO or if it is allowed by other features of quantum mechanics. As evidenced by the debate about the role of ICO in the activation of channel capacities \cite{Channel_ICO,abbott,Guerin_NoiseReduction_DCO, Hler_ICO_Channel_resource}, this is an important question to ask in understanding whether indefinite causality is responsible for any given phenomenon. Indeed, we find a protocol that occurs in a well defined order which allows for this same feature of private detection. To do this, however, an extra instance of Alice's operation seems to be required, a property consistent with other discussions of indefinite versus definite causal orderings \cite{Quantum_Switch_original}. Further, motivated by a contrast in possible eavesdropping strategies, we note that there are hints of some more intriguing differences between quantum encryption in an definite causal world, and an indefinite one. Due to the practical concerns we discuss, we ultimately conclude that ICO may not offer an advantage to QKD in the way considered here. However, we note that this is just an initial step in considering this combination of topics, arguing that there are many interesting future lines of research, both from a foundational and practical standpoint.

% \footnote{There has been work in error correction and privacy amplification without publicly declaring any of the key (that is, the publicly shared information is put through various hash functions) \cite{MaQKD}. However, being based on classical post-processing techniques, this is different from what we find here.}

In Sec.\,\ref{sec:ICKD_Background_Theory} we briefly discuss the general background theory of the two topics of importance in this article: indefinite causal order and quantum key distribution. In Sec.\,\ref{sec:NoEve}, we describe how a key can be distributed between Alice and Bob in an indefinite causal order when no eavesdropper is present. In Sec.\,\ref{sec:OneEve} we introduce a single eavesdropper to gain some intuition of their effects. One eavesdropper location being insufficient to prove the security of this protocol, in Sec.\,\ref{sec:security}, a second and final eavesdropper is introduced, and we explore the security of this protocol by considering a class of individual attacks by the eavesdroppers. In Sec.\,\ref{sec:correlated_and_beyond}, we state our main result: that eavesdroppers, performing {\it any} individual attack, cannot extract any useful information about the shared key without revealing themselves, then outline a roadmap to generalising our security proof to coherent attacks. Following this, in Sec.\,\ref{sec:DefiniteCausalOrder?} we briefly discuss whether this phenomenon is truly a consequence of indefinite causal order, and whether there are any alternate differences between the definite and indefinite cases. Finally, in Sec.\,\ref{sec:Conc}, the findings are summarised along with a discussion of the implementability and practicality of this protocol, as well as possible future lines of research.

\section{Background theory}\label{sec:ICKD_Background_Theory}

\subsection{Indefinite causal order}\label{sec:ICO_Background_Theory}

Suppose two parties, Alice and Bob, hope to act on some state $\rho$ sequentially with the respective operations $\mathcal{A}, \mathcal{B}$ (or more generally, sets of operations, i.e. quantum instruments\footnote{Defined in Appendix\,\ref{sec:instruments_etc}.}) defined using the Kraus operators $\{A_i\},\{B_j\}$, respectively. Normally, at least from a classical perspective, this occurs, as depicted in FIG.\,\ref{fig:DCOvsICO}, in a {\it definite} order: either Alice before Bob,
\begin{equation}
    \rho \to \sum_i A_i \,\rho \,A_i^{\dagger} \to \sum_{i,j} B_j A_i \,\rho \,A_i^{\dagger} B_j^{\dagger},
\end{equation}
or Bob before Alice,
\begin{equation}
    \rho \to \sum_j B_j \,\rho \, B_j^{\dagger} \to \sum_{i,j} A_i B_j \,\rho \, B_j^{\dagger} A_i^{\dagger}.
\end{equation}
In quantum mechanics, however, the order in which Alice and Bob act on $\rho$ can be indefinite - a phenomenon known as \textit{indefinite causal order} (ICO). Take the quantum switch for example\footnote{This is just one example of indefinite causal order, but by far the most understood, and it is what we use throughout this article. For more general discussions, see e.g. Ref.\,\cite{Witnessing_causal_nonseparability}.}, where an extra, control qubit in the state $\omega$ dictates the order in which Alice and Bob act on $\rho$. Much like a classical switch, if we turn it {\it on} and set $\omega = |1\rlangle 1|$, then Alice acts before Bob. Conversely, if we switch it {\it off} and let $\omega = |0\rlangle 0|$, then Bob would go before Alice. However, since $\omega$ is a quantum state, it can be in a superposition of $|0\rangle$ and $|1\rangle$, meaning that Alice and Bob can act on $\rho$ in a {\it controlled superposition} of orders.

Let us write this down mathematically. As mentioned, if the control qubit is in the state $\omega = |1\rlangle 1|$, then Alice acts on the target qubit using $\mathcal{A}$ before Bob acts on it with $\mathcal{B}$, and if $\omega = |0\rlangle 0|$, then $\mathcal{B}$ occurs before $\mathcal{A}$. Following the notation of Ref.\,\cite{Indef_Noise}, we can therefore write the quantum switch as the following operation
\begin{equation}
    \rho \to \mathcal{S}_{\omega} (\mathcal{A},\mathcal{B})(\rho) = \sum_{i,j} S_{ij}\, (\rho \otimes \omega)  \,S_{ij}^{\dagger},
\end{equation}
where the Kraus operators $\{S_{ij}\}$ are defined as
\begin{equation}
    S_{ij} = A_i B_j \otimes |0\rlangle 0| + B_j A_i  \otimes |1\rlangle 1|
\end{equation}
This is depicted in FIG.\,\ref{fig:DCOvsICO}(c).

\begin{figure*}
    \centering
    \includegraphics[width=0.99\textwidth]{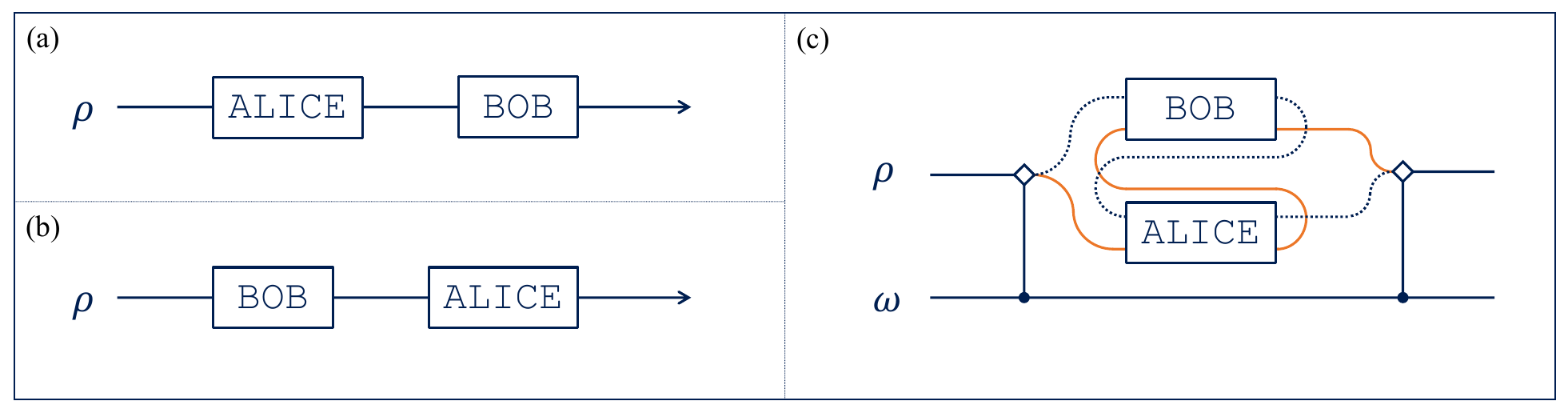}
    \caption{Quantum mechanics allows for more freedom in the ordering of events: (a) Alice can act on a state $\rho$ before Bob, (b) Bob before Alice, or (c) in a superposition of both orders, controlled on some quantum state $\omega$.}
    \label{fig:DCOvsICO}
\end{figure*}

\subsection{Quantum key distribution}

Aside from being in an indefinite causal order with each another, Alice and Bob also like to share private keys with one other to use for various cryptographic tasks. In quantum key distribution (QKD), this is often done (for example in the original implementation: BB84 \cite{bb84}) by having Alice send a key to Bob using encoded quantum systems. To do this, she prepares qubits in states that correspond to the 0s and 1s of the private key and sends them to Bob to be measured. Indeed, in BB84, Alice and Bob respectively prepare and measure, independently and randomly, in one of two mutually unbiased bases. In this work, we will use the Pauli $x$ and $z$-bases: $\{ |0\rangle, |1\rangle \}$ and $\{ |+\rangle, |-\rangle\}$ respectively. If Alice (Bob) prepared (measured) the qubit to be in the state $|0\rangle$ or $|+\rangle$, she (he) will have a corresponding key bit of $0$. Likewise, $|1\rangle, |-\rangle$ corresponds to the key bit $1$. Once Bob has measured the qubit Alice sent him, the two parties publicly discuss which bases they chose. If they chose different bases, there is only a 50\% chance of them agreeing on the corresponding key bit value, so they discard it. If, however, they chose the same basis, when no eavesdroppers are present, Bob's measurement result is guaranteed to correspond to the state prepared by Alice, assuming noiseless and lossless transmission, as we will do throughout. Therefore, Alice and Bob can use the corresponding ordered set of key bit values as their shared key.

The security of this protocol comes from the fact that when an eavesdropper, Eve, intercepts the transmission from Alice to Bob and tries to learn the key bit value being shared, she disturbs the quantum state being sent with non-zero probability\footnote{This is a consequence of quantum measurement disturbance: to learn which state Alice prepared the system in, Eve must measure it. Since she has no idea which of the mutually unbiased bases Alice chose to prepare the state in, she cannot choose a measurement that will always leave the state undisturbed.}. This means that, even if Alice and Bob agree on the basis chosen, there is a non-zero probability that they disagree on the state of the qubit, which implies that there is a chance of an error in Bob's key with respect to Alice's. To detect these errors, Alice and Bob take a subset of their sifted keys and compare them {\it publicly}. Since it has to be done publicly, this subset must subsequently be discarded, regardless of whether errors, and therefore Eve, were detected or not. Let us now see how this protocol can be adapted to an indefinite causal setting.

\section{Quantum key distribution in an indefinite causal order}

\subsection{Indefinite causal key distribution with no eavesdroppers}\label{sec:NoEve}

In BB84, Alice would prepare the qubits to be sent to Bob in a certain state. When considering an indefinite causal ordered scheme, Alice is simultaneously sending and receiving the qubit from Bob (and vice versa). So having one party be the ``preparer" of the state, and the other the ``measurer" makes little sense. To avoid this, we let {\it both Alice and Bob} measure the qubit being used, which, because of how states are updated following projective measurements, allows them to both be the preparer and measurer of the shared qubit. This method has similarities to how a key is generated in protocols like E91 \cite{E91}. Taking this approach, each bit of the key would be the result of projective measurements performed by Alice and Bob on the shared qubit, initially in the state $\rho$, but only when Alice and Bob agree they had performed the {\it same} measurement. Due to the projective nature of these measurements, Alice and Bob would obtain identical measurement outcomes and therefore share an identical key bit (assuming noiseless and lossless channels).

Thinking of the key generation in this way, we can consider a scheme in which a key is distributed in an indefinite causal order. Here, we send a state $\rho$ to two parties, Alice and Bob, in a controlled superposition of two orders: Alice before Bob and Bob before Alice. As shown in FIG.\,\ref{fig:AB}, and as discussed in Sec.\,\ref{sec:ICO_Background_Theory}, this superposition is controlled by the qubit state $\omega$.
Alice and Bob then both make a random choice to measure either in the Pauli $z$-basis: $\{ |0\rangle, |1\rangle \}$, or $x$-basis: $\{|+\rangle, |-\rangle\}$. We can think of Alice and Bob as acting on the state with quantum channels $\mathcal{A}, \mathcal{B} \equiv \mathcal{A}$ respectively\footnote{We denote Alice and Bob's channels using different letters for clarity, but they are identical and can be interchanged whenever it is useful.}, defined by the respective sets of Kraus operators $\{A_i\}, \{B_j\}$, such that
\begin{subequations}\label{Kraus}
\begin{alignat}{4}
    A_0 &\equiv B_0 = \frac{1}{\sqrt{2}}|0\rangle\langle 0|,\\
    A_1 &\equiv B_1 = \frac{1}{\sqrt{2}}|1\rangle\langle 1|,\\
    A_+ &\equiv B_+ = \frac{1}{\sqrt{2}}|+\rangle\langle +|,\\
    A_- &\equiv B_- = \frac{1}{\sqrt{2}}|-\rangle\langle -|.
\end{alignat}
\end{subequations}
Here, the factors of $1/\sqrt2$ arise because we are assuming Alice and Bob are both equally likely to measure in the $x$ or $z$-basis\footnote{This is not always the case in practical QKD, often, one basis is taken to be heavily biased over the other \cite{Asymmetric}.}. It should be made clear that Alice and Bob are not just putting $\rho$ through some quantum channel, they are indeed performing the stated measurements. They could, for example, store their measurement results in a four dimensional ancillary register $R$ (available only in their respective laboratories) initially in the state $|m_0\rangle^R$. The corresponding Kraus operators that would achieve this would have the form $A'_i \equiv B_i' = |i\rangle\langle i| \otimes |m_i\rangle^R \langle m_0|^R/\sqrt{2}$, where $|m_i\rangle$ encode the four possible measurement outcomes $i\in \{0,1,+,-\}$ in orthogonal states: $\langle m_i | m_j\rangle = \delta_{ij}$. Having said this, since these ancillary systems factor out, we can take $A_i,B_i$ to have the form given in Eq.\,(\ref{Kraus}).
\begin{figure}
    \centering
    \includegraphics[width=0.48\textwidth]{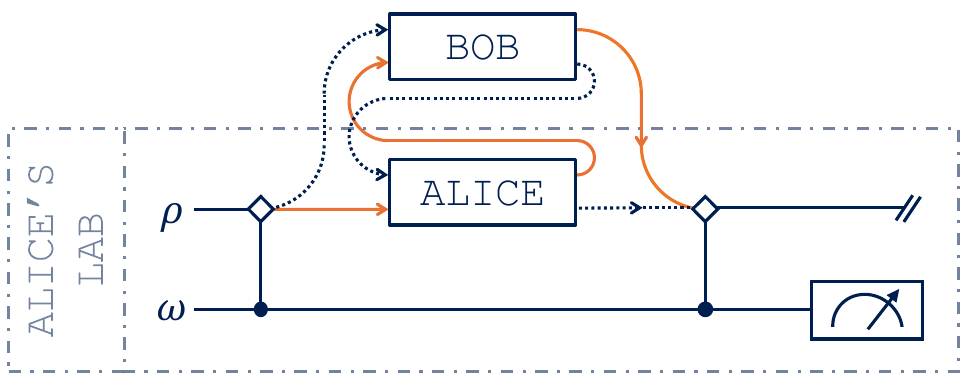}
    \caption{Indefinite causal key distribution with no eavesdroppers. A key is shared between Alice and Bob by sending a state $\rho$ to them in a superposition of orders controlled by the state $\omega$. Alice and Bob perform projective measurements randomly in either the Pauli $x$ or $z$-basis. After discarding cases in which Alice and Bob measured in different bases, they are left with identical keys. Regardless of the initial state $\omega$ of the control qubit, $\omega$ never changes when there are no eavesdroppers, a phenomenon we see not to be true when an eavesdropper is introduced.}
    \label{fig:AB}
\end{figure}

Following their measurements, Alice and Bob then publicly discuss the basis they chose for each measurement and only keep the measurement outcomes in which they measured $\rho$ in the same basis. Assuming no errors occur between Alice and Bob's measurements, their keys, made up of the measurement outcomes they kept, should be identical. In what follows, similarly to what we discussed earlier, a measurement outcome of $0$ and $+$ will correspond to a $0$ in the key. Likewise, $1$ and $-$ correspond to a $1$ in the key.

Let's see in more detail out what happens to the state $\rho$ when it is put through the setup in FIG.\,\ref{fig:AB}. Repeating what was written in Sec.\,\ref{sec:ICO_Background_Theory}, the channel that $\rho$ goes through, is given by
\begin{equation}\label{S_Channel}
    \mathcal{S}_{\omega}(\mathcal{A},\mathcal{B})(\rho) = \sum \limits_{i,j\in \mathfrak{I}} S_{ij} \rho \otimes \omega S_{ij}^{\dagger},
\end{equation}
where
\begin{equation}
    S_{ij} = A_i B_j \otimes |0\rangle\langle 0| + B_j A_i \otimes |1\rangle\langle 1|,
\end{equation}
and we define the set containing the Kraus operator indices by 
\begin{equation}
    \mathfrak{I} := \{0,1,+,-\}.
\end{equation}
After some algebra, it follows that Eq.\,(\ref{S_Channel}) can be rewritten as follows:
\begin{multline}\label{Channel_no_eavesdroppers}
    \mathcal{S}_{\omega} (\mathcal{A}, \mathcal{B})(\rho) = \frac14 \sum \limits_{i,j\in \mathfrak{I}} \big(\{A_i, B_j\} \rho \{A_i, B_j\}^{\dagger} \otimes \omega 
    \\+ [A_i, B_j] \rho [A_i, B_j]^{\dagger} \otimes \sigma_z\omega \sigma_z\big),
\end{multline}
where $\sigma_z$ is the $z$ Pauli operator.

Now, recall that, after public discussion, Alice and Bob only keep the cases in which they performed a measurement in the same basis. Therefore, following this discussion, the state becomes
\begin{multline}
    \mathcal{S}_{\omega} (\mathcal{A}, \mathcal{B})(\rho) \to \frac12 \sum \limits_{\mathfrak{B} \in \mathsf{C}} \sum \limits_{i,j\in \mathfrak{B}} \big( \{A_i, B_j\} \rho \{A_i, B_j\}^{\dagger} \otimes \omega 
    \\+ [A_i, B_j] \rho [A_i, B_j]^{\dagger} \otimes \sigma_z \omega \sigma_z \big),
\end{multline}
where the prefactor is found by requiring normalisation, $\mathsf{C} = \{\{0,1\}, \{+,-\}\}$, and $\mathfrak{B}$ labels the elements of $\mathsf{C}$. Noting the form of $A_k, B_k$ given in Eq.\,(\ref{Kraus}), the terms in these sums have the following properties
\begin{align}
    \begin{aligned}
    \{A_i,B_j\} &= \sqrt2 A_i\delta_{ij},\\
    [A_i, B_j] &= 0,
    \end{aligned}
\end{align}
for all $i,j$ from the same basis, where $\delta_{ij}$ is the Kronecker delta. This confirms that Alice and Bob must agree on their measurement outcomes, meaning we can successfully share a key in an ICO. Overall, we have that
\begin{equation}\label{NoEveFinalState}
    \mathcal{S}_{\omega} (\mathcal{A}, \mathcal{A})(\rho) \to \sum \limits_{i\in \mathfrak{I}} A_i \rho A_i^{\dagger} \otimes \omega.
\end{equation}
So, when there are no eavesdroppers present, the control state $\omega$ stays in its original state and this situation is ultimately no different from when the causal order is definite. Let us introduce an eavesdropper to see what changes.

\subsection{Introducing an eavesdropper}\label{sec:OneEve}

\begin{figure}
    \centering
    \includegraphics[width=0.48\textwidth]{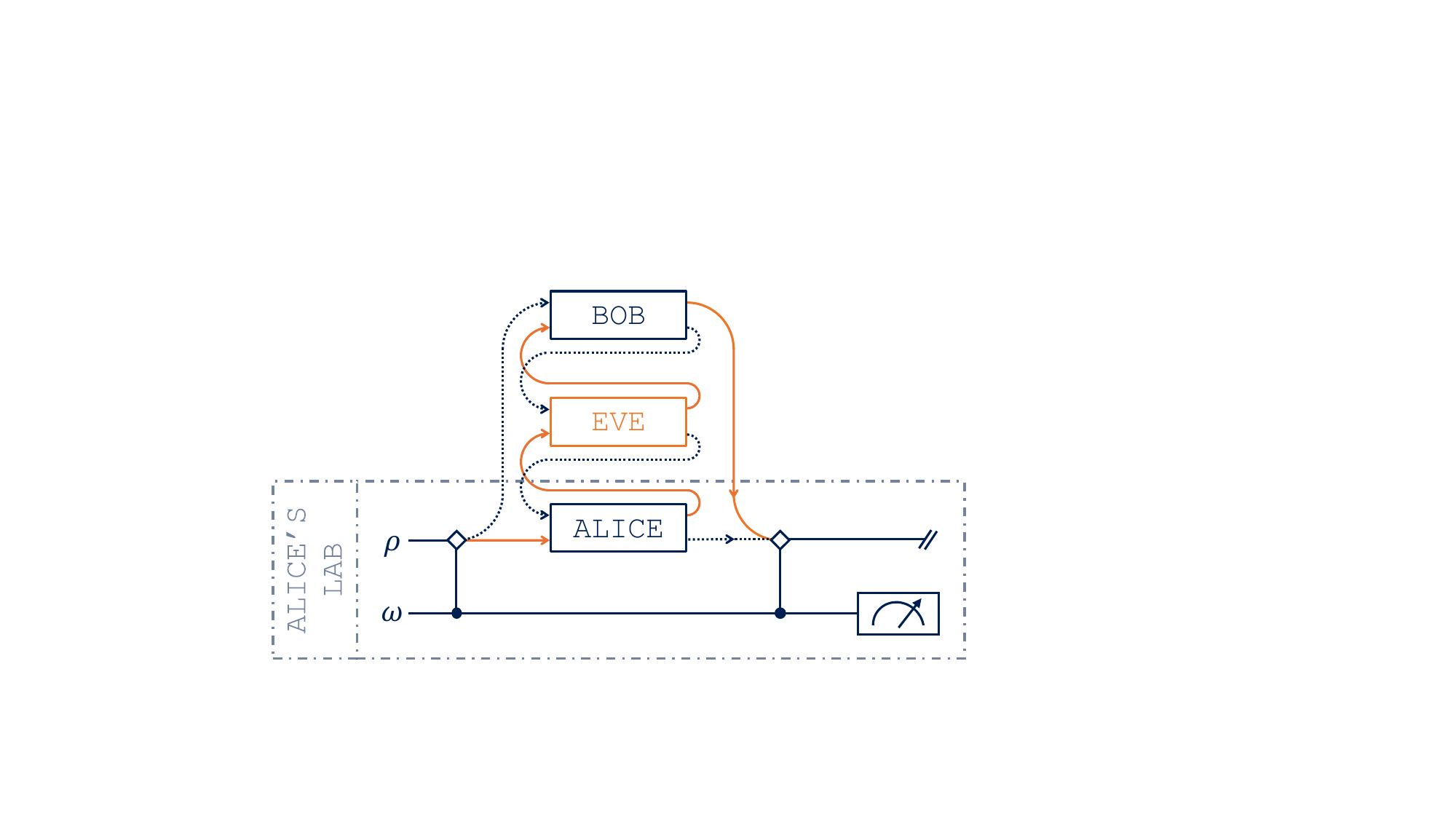}
    \caption{Indefinite causal key distribution with a single eavesdropper, Eve, between Alice and Bob.}
    \label{fig:BEA}
\end{figure}

Notice that, unlike in BB84, there are multiple places an eavesdropper can reside (see FIG.\,\ref{fig:YBEA}). Having said this, to gain some intuition about the effects of eavesdroppers, let us first consider introducing just a single eavesdropper, Eve. In the original BB84 protocol, Eve interacts with the qubit when in transit {\it between} Alice and Bob, regardless of who the sender/receiver was:
\begin{subequations}
    \begin{alignat}{4}
        \rho &\xrightarrow[\text{sends}]{\text{Alice}} \sum_{i,j,k} B_k E_j A_i \,\rho \,A_i^{\dagger} E_j^{\dagger} B_k^{\dagger},\\
        \rho &\xrightarrow[\text{sends}]{\text{Bob}} \sum_{i,j, k} A_i E_j B_k \,\rho \, B_k^{\dagger} E_j^{\dagger} A_i^{\dagger}.
    \end{alignat}
\end{subequations}
Here, for simplicity\footnote{More generally, Eve has access to a quantum instrument. This is discussed in more detail in Appendix\,\ref{sec:Correlated}.}, denote the channel corresponding to Eve's measurement by $\mathcal{E}$, defined by the Kraus operators $\{E_i\}$.

Now, suppose Alice and Bob carry out the indefinite causal QKD method described in the previous subsection and Eve takes motivation from the definite causal case, setting herself up in between Alice and Bob. This scenario is depicted in FIG.\,\ref{fig:BEA}.
As before, allowing a state $\rho$ to be acted on by Alice, Eve and Bob in an indefinite causal order controlled by $\omega$, the channel $\rho$ passes through is given by
\begin{equation}
    \mathcal{S}_{\omega}(\mathcal{A}, \mathcal{E}, \mathcal{B})(\rho) = \sum \limits_{i,j,k} S_{ijk} \rho \otimes \omega S_{ijk}^{\dagger},
\end{equation}
where
\begin{align}\label{S_ijk_orig}
    S_{ijk} :=& A_i E_j B_k \otimes |0\rangle\langle0| + B_k E_j A_i \otimes |1\rangle\langle1|\nonumber \\
    =& \frac12 \{A_i, E_j, B_k\} \otimes \mathbbm1 + \frac12 [A_i, E_j, B_k] \otimes \sigma_z
\end{align}
such that
\begin{subequations}
\begin{alignat}{4}
    \{A_{i}, E_{j}, B_{k}\} &:= A_i E_j B_k + B_k E_j A_i,\\
    [A_{i}, E_{j}, B_{k}] &:= A_i E_j B_k - B_k E_j A_i.
\end{alignat}
\end{subequations} 
After some algebra, and following the basis comparison,
\begin{widetext}
\begin{equation}\label{keepState}
    \mathcal{S}_{\omega} (\mathcal{A}, \mathcal{E}, \mathcal{B})(\rho) \to  \rho_{\text{sifted}} =\frac12\sum \limits_j \sum \limits_{\mathfrak{B} \in \mathsf{C}} \sum \limits_{i,k\in \mathfrak{B}} \left(\{A_{i}, E_{j}, B_{k}\} \rho \{A_{i}, E_{j}, B_{k}\}^{\dagger} \otimes \omega + [A_{i}, E_{j}, B_{k}] \rho [A_{i}, E_{j}, B_{k}]^{\dagger} \otimes \sigma_z \omega \sigma_z\right).
\end{equation}
\end{widetext}

From this, we can see that, as before, the $\omega$ terms survive. But more interestingly, notice that the $\sigma_z \omega \sigma_z$ terms \textit{can} survive too. For example, suppose Alice and Bob measure in the $z$-basis and Eve measures in the $x$-basis, then it is possible for Alice to obtain an outcome of 0, and Bob an outcome of 1. This combination allows for $[A_0, E_{\pm}, B_1] \neq 0$.

We may therefore hypothesise that if Eve attempts to extract information about the state when in between Alice and Bob, she induces a nonzero $\sigma_z \omega \sigma_z$ term. So, if we were to let $\omega = |+\rangle\langle +|$ (and therefore $\sigma_z \omega \sigma_z = |-\rangle\langle-|$), if someone were to perform the measurement $\{ |+\rangle\langle +|, |-\rangle\langle -|\}$ on the control state $\omega$, and obtain an outcome of $-$, they could conclude that there was an eavesdropper in between Alice and Bob. Moreover, since Alice can keep and measure the control qubit in her own lab, this would mean that no subset of the distributed key need be publicly compared and discarded to determine the presence of Eve. Let us now explore how robust this hypothesis is. %Further, in this case, when Eve is detected, it turns out that Alice and Bob's results are necessarily different. This means half of all Eve-induced errors can be corrected just by flipping either Alice or Bob's bit whenever Eve is detected. The remaining Eve-induced errors live in the $\omega$ terms and can removed using the normal information reconciliation and privacy amplification methods \cite{PrivAmp}.

\section{Security against individual attacks}\label{sec:security}

As alluded to earlier, there are two possible positions eavesdroppers can exist: in between, or outwith Alice and Bob's encoding operations. Let's see what happens when two eavesdroppers, Eve and Yves, are introduced. In particular, in this work, we consider {\it individual} attacks where Eve and Yves act on each distributed state \textit{separately} \cite{QKDreview} and we leave more general attacks for future work (although we outline how one might analyse them in the following section). There being more than one eavesdropper allows for cooperative strategies, the most general of which utilise both quantum and classical correlations between Eve and Yves's operations. Possible correlations include everything from shared entangled ancilla states, to quantum channels, to more general indefinite causal structures and can be described mathematically by Eve and Yves sharing a {\it process matrix} $W^{\tilde{E} \tilde{Y}}$, as depicted in FIG.\,\ref{fig:YBEA}. We discuss this scenario, which we will often call the ``fully correlated" case, in more depth in Appendix\,\ref{sec:Correlated}.

\begin{protocol}
    \centering
    \begin{mdframed}[style=mystyle]
    \begin{enumerate}
    \item Alice prepares the state $\rho = |y\rlangle y|$, where $|y\rangle = (|0\rangle + i |1\rangle)/\sqrt2$, to be distributed along with a control qubit state $\omega = |\!+\rlangle +|$ that remains in her lab.
    \item The state $\rho$ is distributed to Alice and Bob's measuring devices in an indefinite causal order, controlled on $\omega$. 
    \item Alice and Bob non-destructively measure $\rho$ in either the $x$ or $z$-basis, chosen randomly with probability $1/2$. Measurement outcomes $0, +$ correspond to a key bit $0$, and outcomes $1, - $ correspond to a key bit of $1$.
    \item Alice and Bob compare the bases they chose and only keep the cases in which they agree.
    \item For each (sifted) state $\rho$ distributed, Alice measures the corresponding control qubit state $\omega$ in the $x$-basis. If an outcome $+$ is obtained, carry on. If an outcome $-$ is found, Alice concludes eavesdroppers are present after which, she either aborts the key distribution, or she and Bob go ahead with privacy amplification and error correction (not discussed in this work).
    \end{enumerate}
    \end{mdframed}
    \caption{Summary of the proposed indefinite causal key distribution protocol.}
    \label{protocol:protocol}
\end{protocol}

In this section, we consider a subclass of these individual eavesdropping strategies where Eve and Yves's ancilliary systems (that they use to aid in their attack) are separable but jointly measurable, and we prove the security of our protocol (which is summarised in Protocol\,\ref{protocol:protocol}) against them. Investigation into the full security proof of this protocol is beyond the scope of this work. This being said, the attacks considered in this section provide us with some useful understanding about the security of this protocol. The methods used follow closely those used in \cite{LM05} where the authors consider a two-way deterministic quantum communications protocol in which quantum states are sent back and forth between Alice and Bob. Although they are working in a definite causal order, in this protocol, the eavesdropper has access to the state at two different points, which is why the same techniques used can be applied to our protocol.

\subsection{Problem setup}

Let $S$ denote the distributed qubit, initially prepared in the state $\rho$. In the scenario we consider, we assume that Alice and Bob would like their shared key to contain, on average, equal numbers of 0s and 1s. Therefore, we can make a natural choice for the initial state of $S$: $\rho = |y\rlangle y|$, where $|y\rangle = (|0\rangle + i |1\rangle)/\sqrt2$. We assume the preparation of this state can be done securely within Alice's lab, perhaps via a Pauli-$y$ measurement to disentangle it from any adversaries. Taking this approach, we'd assume such a measurement device is trusted, as we do with all of the measurement devices in this protocol. 

Now, if Alice is in the lab in which the state $\rho$ is created and where $\omega$ resides, there are two places eavesdroppers, who we call Eve and Yves, can be located. This setup is shown in FIG.\,\ref{fig:YBEA}. In addition to the state $\rho$ being sent between Alice and Bob, Eve and Yves also have access to independent ancilliary quantum systems $\tilde{E},\tilde{Y}$ (respectively) initially in the states $\varepsilon := |\varepsilon\rangle\langle \varepsilon|$ and $\eta := |\eta\rangle\langle \eta|$ respectively\footnote{Note that we are not here considering the fully correlated case of Eve and Yves sharing an arbitrary process matrix $W^{\tilde{E}\tilde{Y}}$.}. Eve and Yves perform the respective unitaries $U_E^{(S\tilde{E})}=:U_E$, $U_Y^{(S\tilde{Y})} =:U_Y$ on the joint space of the distributed state $\rho$ and the spaces of their respective ancillae $\varepsilon, \eta$. Following this, the eavesdroppers perform some joint measurement on the ancillae to try and gain some information about Alice and Bob's shared key. Note that Eve and Yves do, therefore, cooperate in this scenario, although not in the most general way possible.
\begin{figure}
    \centering
    \includegraphics[width=0.48\textwidth]{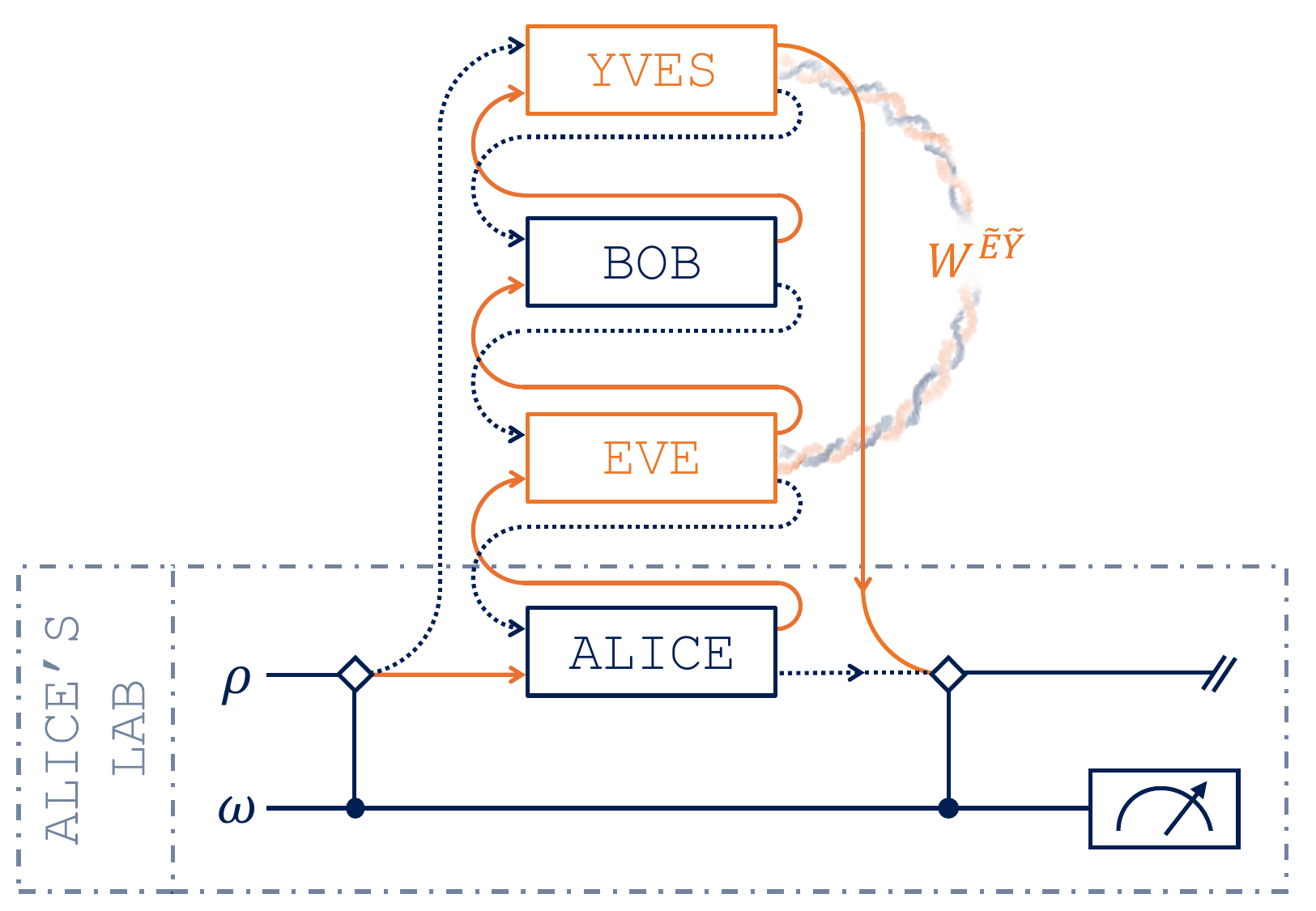}
    \caption{Indefinite causal quantum key distribution with two eavesdroppers Eve and Yves. The fact there are now two eavesdroppers means correlated attacks are possible. All possible correlations are described mathematically by a process matrix $W^{\tilde{E}\tilde{Y}}$ shared by Eve and Yves.}
    \label{fig:YBEA}
\end{figure}

For $k\in\{0,1,+,-\}$, the eavesdroppers' unitaries are performed most generally as follows \cite{LM05}:
\begin{subequations}\label{GeneralU_E&Y}
\begin{alignat}{4}
    U_E^{(S\tilde{E})} |k\rangle^{(S)} |\varepsilon\rangle^{(\tilde{E})} &= |k\rangle|\varepsilon_{kk}\rangle +|\bar{k}\rangle |\varepsilon_{k \bar{k}}\rangle,\\
    U_Y^{(S\tilde{Y})} |k\rangle^{(S)} |\eta\rangle^{(\tilde{Y})} &= |k\rangle|\eta_{kk}\rangle +|\bar{k}\rangle |\eta_{k \bar{k}}\rangle,
\end{alignat}
\end{subequations}
where $|\varepsilon_{mn}\rangle$ and $|\eta_{mn}\rangle$ are, in general, unnormalised and non-orthogonal, and $\bar{k}$ is taken to mean ``not $k$". For simplicity, we follow the approach of \cite{wlog, LM05} and take $\langle \varepsilon_{kl}|\varepsilon_{mn}\rangle, \langle \eta_{kl}| \eta_{mn}\rangle \in \mathbb{R} ~\forall k,l,m,n\in \{0,1,+,-\}$ which allows us to set\footnote{Let $k\in \{0,1\}$ first. Since the system $S$ is two dimensional, by the Schmidt decomposition \cite{nielsenAndChuang}, we can write $U_E |k\rangle |\varepsilon\rangle = |\psi_k\rangle|\xi_{kk} \rangle +|\psi_{\bar{k}}\rangle |\xi_{k\bar{k}}\rangle$ such that $\langle\psi_k | \psi_{\bar{k}} \rangle = 0 = \langle \xi_{kk} | \xi_{k\bar{k}}\rangle$. Noting that there exist $\theta_k\in [0,\pi], \phi_k \in [0,2\pi)$ such that $|\psi_k\rangle = \cos (\theta_k/2)|0\rangle + e^{i\phi_k}\sin(\theta_k/2)|1\rangle, ~ |\psi_{\bar{k}}\rangle = -\sin (\theta_k/2)|0\rangle + e^{i\phi_k}\cos(\theta_k/2)|1\rangle$, we can rewrite our unitary operation as $U_{E}|k\rangle |\varepsilon\rangle = |k\rangle |\varepsilon_{kk}\rangle + |\bar{k}\rangle | \varepsilon_{k\bar{k}}\rangle$ such that $|\varepsilon_{kk}\rangle = \cos(\theta_k/2)|\xi_{kk}\rangle - \sin(\theta_k/2)|\xi_{k\bar{k}}\rangle, ~ |\varepsilon_{k\bar{k}}\rangle = e^{i\phi_k}\big[\sin(\theta_k/2)|\xi_{kk}\rangle + \cos(\theta_k/2)|\xi_{k\bar{k}}\rangle\big]$. A little algebra confirms that $\langle \varepsilon_{kk}|\varepsilon_{k\bar{k}}\rangle = 0$. Following the same argument with $k\in \{+,-\}$ leads to $\langle \varepsilon_{kk}|\varepsilon_{\bar{k}k}\rangle = 0$. Identical reasoning applies to the $U_Y$ case.} $\langle \varepsilon_{kk} | \varepsilon_{k\bar{k}}\rangle = \langle \varepsilon_{kk} | \varepsilon_{\bar{k}k}\rangle = 0 = \langle \eta_{kk} | \eta_{k\bar{k}}\rangle = \langle \eta_{kk} | \eta_{\bar{k}k}\rangle , \forall k\in \{0,1, +, -\}$. Note that we, from now on, will drop the superscripts unless it is unclear which space is which. Note also, that\footnote{To see this explicitly, we perform $U_E |\pm\rangle |\varepsilon\rangle = U_E(|0\rangle|\varepsilon\rangle \pm |1\rangle |\varepsilon\rangle)/\sqrt2$ using Eq.\,(\ref{GeneralU_E&Y}a). This results in $(|0\rangle |\varepsilon_{00}\rangle + |1\rangle|\varepsilon_{01}\rangle \pm |1\rangle |\varepsilon_{11}\rangle \pm |0\rangle |\varepsilon_{10}\rangle)/\sqrt2 = \big[ |0\rangle (|\varepsilon_{00}\rangle \pm |\varepsilon_{10}\rangle) + |1\rangle(|\varepsilon_{01}\rangle \pm |\varepsilon_{11}\rangle)\big]/\sqrt2$, which can be rewritten (ignoring the global phase $\pm1$) as $\big[|\pm\rangle (|\varepsilon_{00}\rangle \pm |\varepsilon_{01}\rangle \pm |\varepsilon_{10}\rangle +|\varepsilon_{11}\rangle) + |\mp\rangle (|\varepsilon_{00}\rangle \mp |\varepsilon_{01}\rangle \pm |\varepsilon_{10}\rangle +|\varepsilon_{11}\rangle)\big]/2.$ The form of $|\varepsilon_{\pm \pm}\rangle$, $|\varepsilon_{\pm \mp}\rangle$ follow from this. A similar argument can be used to derive the various $|\eta_{mn}\rangle$ expressions.}
\begin{subequations}
    \begin{alignat}{4}
        |\varepsilon_{\pm \pm}\rangle &= \frac12(|\varepsilon_{00}\rangle \pm |\varepsilon_{01}\rangle \pm |\varepsilon_{10}\rangle + |\varepsilon_{11}\rangle),\\
        |\varepsilon_{\pm \mp}\rangle &= \frac12(|\varepsilon_{00}\rangle \mp |\varepsilon_{01}\rangle \pm |\varepsilon_{10}\rangle - |\varepsilon_{11}\rangle),
    \end{alignat}
\end{subequations}
and
\begin{subequations}
    \begin{alignat}{4}
        |\eta_{\pm \pm}\rangle &= \frac12(|\eta_{00}\rangle \pm |\eta_{01}\rangle \pm |\eta_{10}\rangle + |\eta_{11}\rangle),\\
        |\eta_{\pm \mp}\rangle &= \frac12(|\eta_{00}\rangle \mp |\eta_{01}\rangle \pm |\eta_{10}\rangle - |\eta_{11}\rangle).
    \end{alignat}
\end{subequations}
As our input state is $|y\rangle,$ the following will also be useful to our derivation:
\begin{subequations}
    \begin{alignat}{4}
        |\eta_{y y}\rangle &= \frac12(|\eta_{00}\rangle - i|\eta_{01}\rangle + i|\eta_{10}\rangle + |\eta_{11}\rangle),\\
        |\eta_{y \bar{y}}\rangle &= \frac12(|\eta_{00}\rangle + i |\eta_{01}\rangle + i|\eta_{10}\rangle - |\eta_{11}\rangle).
    \end{alignat}
\end{subequations}

When $k\in \{0,1\}$, we define $\langle \varepsilon_{kk}| \varepsilon_{kk} \rangle = F$, $\langle \varepsilon_{k\bar{k}}| \varepsilon_{k\bar{k}} \rangle = D $ and $\langle \eta_{kk}| \eta_{kk} \rangle = F'$, $\langle \eta_{k\bar{k}}| \eta_{k\bar{k}} \rangle = D'$, which, as mentioned, are taken to be positive real numbers. These values relate to the probability that Eve and Yves leave the distributed state unaffected. In order to ensure unitarity, 
\begin{subequations}\label{unitary_conditions}
    \begin{alignat}{4}
        F+D &= 1 = F' + D',\\
        \langle\varepsilon_{00}| \varepsilon_{10}\rangle + \langle\varepsilon_{01}| \varepsilon_{11}\rangle &= 0 = \langle\eta_{00}| \eta_{10}\rangle + \langle\eta_{01}| \eta_{11}\rangle.
    \end{alignat}
\end{subequations}
Since $|\varepsilon_{kl}\rangle, |\varepsilon_{\bar{k}\bar{l}}\rangle$ are generally non-orthogonal (likewise for $|\eta_{kl}\rangle$), we take
\begin{subequations}
    \begin{alignat}{4}
        \langle \varepsilon_{00} | \varepsilon_{11}\rangle &= F\cos{x},\\
        \langle \varepsilon_{01} | \varepsilon_{10}\rangle &= D\cos{y},\\
        \langle \eta_{00} | \eta_{11}\rangle &= F'\cos{x'},\\
        \langle \eta_{01} | \eta_{10}\rangle &= D'\cos{y'},
    \end{alignat}
\end{subequations}
where $x,y,x',y'\in [0,\pi/2]$. We can think of $x,y ~(x',y')$ as dictating the distinguishability between Eve's (Yves's) possible ancilla states.

Similarly to in the previous section, following the basis comparison step, the state of the entire system is updated as follows:
\begin{widetext}
\begin{multline}\label{rhoKeepUnitary}
    \rho\otimes \varepsilon \otimes \eta \otimes \omega \to  \rho_{\text{sifted}} = \frac12 \sum \limits_{\mathfrak{B}\in \mathsf{C}} \sum \limits_{j,k\in \mathfrak{B}} \Big( \{U_Y, B_j, U_E , A_k\} \rho \otimes \varepsilon \otimes \eta \{U_Y, B_j, U_E , A_k\}^{\dagger} \otimes \omega
    \\ \hspace{1.8cm} + [U_Y, B_j, U_E , A_k] \rho \otimes \varepsilon \otimes \eta [U_Y, B_j, U_E , A_k]^{\dagger} \otimes \sigma_z \omega \sigma_z\\
    \hspace{4.5cm} + \{U_Y, B_j, U_E , A_k\} \rho \otimes \varepsilon \otimes \eta [U_Y, B_j, U_E , A_k]^{\dagger} \otimes \omega \sigma_z \\
    +[U_Y, B_j, U_E , A_k] \rho \otimes \varepsilon \otimes \eta \{U_Y, B_j, U_E, A_k\}^{\dagger} \otimes \sigma_z \omega \Big),
\end{multline}    
\end{widetext}
where
\begin{subequations}
    \begin{alignat}{4}
        \{U_Y, B_j, U_E , A_k\} &:= U_Y B_j U_E A_k + A_k U_E B_j U_Y,\\
        [U_Y, B_j, U_E , A_k] &:= U_Y B_j U_E A_k - A_k U_E B_j U_Y.
    \end{alignat}
\end{subequations}
Denoting Eve and Yves's joint strategy using $\mathcal{Z}$, we are now equipped to calculate the following: 
\begin{enumerate}
    \item Minimum probability of detection: $d$.
    \item Eavesdroppers and Alice's mutual information: $H(\mathcal{Z}: \mathcal{A})$.
    \item Alice and Bob's mutual information: $H(\mathcal{A}: \mathcal{B})$.
\end{enumerate}

\subsection{Minimum probability of detection}

Let us first calculate the eavesdropper detection probability. Recall that, in this protocol, this corresponds to measuring the control qubit to be in the state $|-\rangle \langle -|$ given that it was initially prepared in the state $|+\rangle \langle +|$. Therefore, using $\rho = |y\rangle \langle y|$, the probability of detection is given by
\begin{multline}
    P_{\text{detect}} = \frac12 \sum \limits_{\mathfrak{B}\in \mathsf{C}} \sum \limits_{j,k \in \mathfrak{B}} \langle y \varepsilon \eta| [U_Y, B_j, U_E , A_k]^{\dagger}\\ \times 
    [U_Y, B_j, U_E , A_k] |y \varepsilon \eta \rangle,
\end{multline}
where $|y\varepsilon \eta\rangle := |y\rangle^{(S)} \otimes |\varepsilon\rangle^{(\tilde{E})} \otimes |\eta\rangle^{(\tilde{Y})}$. Now, noting that,
\begin{subequations}
    \begin{alignat}{4}
        U_Y B_j U_E  A_k |y \varepsilon \eta \rangle &= \frac{\langle k| y\rangle}{2} \big( \delta_{jk} |\varepsilon_{kk}\rangle + \delta_{j\bar{k}} |\varepsilon_{k\bar{k}}\rangle \big) \nonumber\\ &\hspace{1.8cm}\big( |j\rangle |\eta_{jj}\rangle + |\bar{j}\rangle |\eta_{j\bar{j}}\rangle \big),\\
        A_k U_E B_j U_Y |\psi \varepsilon \eta \rangle &= \frac12 |k\rangle \big(\delta_{kj} |\varepsilon_{jj}\rangle + \delta_{k\bar{j}} |\varepsilon_{j\bar{j}}\rangle \big)\nonumber\\ &\hspace{1.05cm}\big( \langle j|y\rangle |\eta_{yy}\rangle + \langle j|\bar{y}\rangle|\eta_{y\bar{y}}\rangle\big),
    \end{alignat}
\end{subequations}
where, in the first equation, the order of the qubits was changed for convenience, we can calculate $P_{\text{detect}}$ to be
\begin{align}
\begin{aligned}
    P_{\text{detect}} = \frac12 - \frac18\big[&FF'(3 + \cos x \cos x')  \\
    & + DD'(1 + 3\cos y \cos y')  \\
    &\hspace{1cm} + FD'(\cos x + \cos y') \\
    &\hspace{1.5cm} + DF'(\cos y + \cos x') \big].
\end{aligned}
\end{align}
Recalling that $D = 1- F$ and $D' = 1- F'$, and minimising $P_{\text{detect}}$ over $F,F'$, we find there are two\footnote{There are actually 4 possibilities, but the two options not written explicitly: $F = 0 = D'$ and $F = 1 = D'$, result in larger detection probabilities than the $F = 1 = F'$ case $\forall x,x',y,y'$.} possibilities when minimising the probability of detection $d := P_{\text{detect}}^{\text{min}}$. Note that we can take this approach since $D,F,D',F'$ can be chosen {\it independently} of $x,x',y,y'$.

\vspace{3mm}
\noindent {\it Option 1:} $F = 1 = F'$. Here,
\begin{equation}\label{SecurityMinDetection}
    d = \frac18 (1 - \cos x \cos x').
\end{equation}

\vspace{3mm}
\noindent {\it Option 2:} $F = 0 = F'$. Here,
\begin{equation}
    d = \frac38 (1 - \cos y \cos y').
\end{equation}

Since there are values of $x,y,x',y'$ such that each option is smaller, we must consider both options when calculating the various mutual information values.

\subsection{Eavesdroppers - Alice mutual information}

With the setup we're considering, after Eve and Yves have both carried out their unitaries, they perform some joint measurement on both of their ancillae that best distinguishes between a 0 and a 1 in Alice's key. In order to do this, they should utilise all public information, which means waiting for Alice to reveal her basis choice before choosing which joint measurement to perform on their ancillae. Therefore, in order to find the maximum mutual information (subject to a minimal probability of detection), two optimal measurements must be constructed: one to distinguish $\{\Psi_{0}^{\tilde{E}\tilde{Y}}, \Psi_{1}^{\tilde{E}\tilde{Y}}\}$, and one to distinguish $\{\Phi_{+}^{\tilde{E}\tilde{Y}}, \Phi_{-}^{\tilde{E}\tilde{Y}}\}$ which are the possible states of the eavesdroppers' ancillae when the $z$ and $x$-bases were chosen by Alice and Bob, respectively.

For both cases, the states to be distinguished are found using
\begin{subequations}
    \begin{alignat}{4}
        \Psi_{l}^{\tilde{E}\tilde{Y}} &= \frac{1}{\mathcal{N}} \Tr_{S,C} \big( \rho_{\text{sifted}}\big|_{a=l\in\{0,1\}} \big),\\
        \Phi_{l}^{\tilde{E}\tilde{Y}} &= \frac{1}{\mathcal{N}} \Tr_{S,C} \big( \rho_{\text{sifted}}\big|_{a=l \in \{+,-\}} \big),
    \end{alignat}
\end{subequations}
where $\mathcal{N}$ is a normalisation constant, and $S,C$ indicate that the trace is being carried out over the distributed and control qubit respectively. Further, $\rho_{\text{sifted}}|_{a=l\in \mathfrak{B}}$ denotes the terms in $\rho_{\text{sifted}}$ [Eq.\,(\ref{rhoKeepUnitary})] in which Alice's measurement outcome is $l \in \mathfrak{B}$. Let's consider the two options we found in the previous subsection.

%\vspace{3mm}
\subsubsection*{Option 1a: $F = 1 = F', l \in \{0,1\}$}

In this first case of Option 1 where $l\in \{0,1\}$,
\begin{equation}
    \Psi_{l}^{\tilde{E}\tilde{Y}} = \varepsilon_{ll} \otimes \eta_{ll}.
\end{equation}
Note the absence of any $|\varepsilon_{l\bar{l}}\rangle, |\eta_{l\bar{l}}\rangle$ terms since $D = D' =0$. Now, as mentioned earlier, in order to maximise the mutual information between Alice and the eavesdroppers, we must find the optimal measurement that distinguishes $\{ \Psi_{0}^{\tilde{E}\tilde{Y}}, \Psi_{1}^{\tilde{E}\tilde{Y}}\}$. We can do this by noticing that these states can be written as $\Psi_{l}^{\tilde{E}\tilde{Y}}=|\Psi_{l}^{\tilde{E}\tilde{Y}}\rangle\langle \Psi_{l}^{\tilde{E}\tilde{Y}}|$, such that
\begin{multline}\label{Pure_state_in_basis}
    |\Psi_{l}^{\tilde{E}\tilde{Y}}\rangle  = \frac{1}{\sqrt2} \bigg[ \sqrt{1 + |\langle \Psi_0^{\tilde{E}\tilde{Y}} | \Psi_{1}^{\tilde{E}\tilde{Y}}\rangle|}|s\rangle 
    \\+ (-1)^l \sqrt{1 - |\langle \Psi_0^{\tilde{E}\tilde{Y}} | \Psi_{1}^{\tilde{E}\tilde{Y}}\rangle|}|t\rangle\bigg],
\end{multline}
where $|s\rangle, |t\rangle$ are some orthonormal vectors in Eve and Yves's shared ancilla space. In our case, $|\langle \Psi_0^{\tilde{E}\tilde{Y}}| \Psi_1^{\tilde{E}\tilde{Y}} \rangle|$ $= \cos x \cos x'$. The optimal measurement to distinguish these states is known to be made up of the following operators \cite{SarahQSD}:
\begin{equation}
    \pi_l = \frac{1}{2}\Big[|s\rangle + (-1)^l |t\rangle \Big]\Big[\langle s| + (-1)^l \langle t|\Big].
\end{equation}

To calculate the mutual information, we use
\begin{align}\label{MI_gen}
    H_{l/\bar{l}}(\mathcal{Z} : \mathcal{A}) &= -\sum \limits_{i\in \{l,\bar{l}\}} P(z_i) \log P(z_i) \nonumber \\
    & \hspace{1cm} - \sum \limits_{j\in\{l,\bar{l}\}} P(a_j) \log P(a_j) \nonumber \\
    & \hspace{1.5cm} + \sum \limits_{i,j\in\{l,\bar{l}\}} P(z_i, a_j) \log P(z_i, a_j),
\end{align}
where the subscript $l/\bar{l}$ is used to highlight that this is the mutual information {\it only} in the $\{|l\rangle, |\bar{l}\rangle\}$-basis case. Here, the outcome $z_i$ corresponds to Eve and Yves performing their joint optimal measurement $\{\pi_k\}$, with outcome $k=i$, and $a_j$ corresponds to Alice measuring $j$, or equivalently (and perhaps more usefully for the calculation), $a_j$ can be thought of as the preparation of the state $\Psi_j^{\tilde{E}\tilde{Y}}$. After some algebra, it turns out that
\begin{equation}\label{ZA_MI}
    H_{0/1}(\mathcal{Z} : \mathcal{A}) = 1 - h\Big[\big(1 + \sqrt{1 - \cos^2 x \cos^2 x'}\big)/2 \Big],
\end{equation}
where $h(q) = -q\log(q) - (1-q) \log (1-q)$ is the binary entropy function \cite{nielsenAndChuang}. Using Eq.\,(\ref{SecurityMinDetection}), this can be rewritten in terms of the minimum detection probability as
\begin{equation}
    H_{0/1}(\mathcal{Z} : \mathcal{A}) = 1 - h\Big[\big(1 + 4\sqrt{d[1-4d]}\big)/2 \Big],
\end{equation}
such that $0 \leq d \leq 1/8$.

\subsubsection*{Option 1b: $F = 1 = F', l \in \{+,-\}$}

Let's now consider the second case of Option 1 in which Alice and Bob work in the $x$-basis. Here,
\begin{equation}
    \Phi_{l}^{\tilde{E}\tilde{Y}} = \frac34 \big(\varepsilon_{ll}\otimes \gamma_l + \varepsilon_{l\bar{l}} \otimes \gamma_{\bar{l}}\big) + \frac14 \big(\varepsilon_{\bar{l}\bar{l}} \otimes \gamma_{\bar{l}} + \varepsilon_{\bar{l}l}\otimes \gamma_{l}\big),
\end{equation}
where $\gamma_l = |\gamma_l \rlangle \gamma_l|$ such that
\begin{align}
        |\gamma_l\rangle &= \frac{1}{\sqrt{2}}\left( |\eta_{00}\rangle + i(-1)^l|\eta_{11}\rangle\right) \nonumber\\
        &\cong \frac{1}{\sqrt{2}}\left( \sqrt{1 + \cos x'} |u\rangle + (-1)^l \sqrt{1- \cos x'} |v\rangle \right),
\end{align}
where, similarly to Eq.\,(\ref{Pure_state_in_basis}), $\{|u\rangle, |v\rangle\}$ is some orthonormal basis on the span of $|\eta_{00}\rangle, |\eta_{11}\rangle$, and the ``$\cong$" denotes equivalence up to a global phase.

The optimal measurement to distinguish $\Phi^{\tilde{E}\tilde{Y}}_+, \Phi^{\tilde{E}\tilde{Y}}_-$ is made up of the projectors $\xi_+, \xi_-$ onto the positive and negative eigenspaces of $\Phi^{\tilde{E}\tilde{Y}}_+ -  \Phi^{\tilde{E}\tilde{Y}}_-$ \cite{SarahQSD}. Since $\langle \varepsilon_{\pm \pm }| \varepsilon_{\pm \mp}\rangle = 0$, these are respectively given by
\begin{subequations}
    \begin{alignat}{4}
        \xi_+ &= \frac{2}{1+\cos x}\varepsilon_{++} \otimes \pi_+ + \frac{2}{1 - \cos x}\varepsilon_{+-} \otimes \pi_-,\\
        \xi_- &= \frac{2}{1+\cos x}\varepsilon_{--} \otimes \pi_- + \frac{2}{1 - \cos x}\varepsilon_{-+} \otimes \pi_+,
    \end{alignat}
\end{subequations}
where
\begin{equation}
    \pi_{\pm} = \frac{1}{2}(|u\rangle \pm |v\rangle)(\langle u| \pm \langle v|)
\end{equation}
make up the optimal measurement to distinguish $\gamma_{\pm}$.

So, using Eq.\,(\ref{MI_gen}), it follows that
\begin{equation}\label{eaves_MI_x}
    H_{\pm}(\mathcal{Z}:\mathcal{A}) = 1 - h\left[ \big(2 + \sqrt{1- \cos x \cos x'}\big)/4 \right]
\end{equation}
which can again be written in terms of the minimum detection probability as follows:
\begin{equation}
    H_{\pm}(\mathcal{Z}:\mathcal{A}) = 1 - h\left[ \big(1 + \sqrt{2d}\big)/2 \right]\!.
\end{equation}

\subsubsection*{Option 2a: $F = 0 = F', ~ l\in \{0,1\}$}

In this case,
\begin{equation}
    \Psi_{l}^{\tilde{E}\tilde{Y}} = \frac{1}{2} ( \varepsilon_{01} \otimes \eta_{10} + \varepsilon_{10} \otimes \eta_{01} ),
\end{equation}
which means that $\Psi_{0}^{\tilde{E}\tilde{Y}} = \Psi_{1}^{\tilde{E}\tilde{Y}}$. This implies Eve and Yves's best strategy is to just guess. Therefore, the mutual information between the eavesdroppers and Alice is $H_{0/1}(\mathcal{Z}: \mathcal{A}) = 0$ when $F = 0 = F'$ and $l\in \{0,1\}$.

\subsubsection*{Option 2b: $F = 0 = F', ~ l\in \{+, -\}$}

In this case,
\begin{equation}
    \Phi_{l}^{\tilde{E}\tilde{Y}} = \frac34 \big(\varepsilon_{ll}\otimes \tau_l + \varepsilon_{l\bar{l}} \otimes \tau_{\bar{l}}\big) + \frac14 \big(\varepsilon_{\bar{l}\bar{l}} \otimes \tau_{\bar{l}} + \varepsilon_{\bar{l}l}\otimes \tau_{l}\big),
\end{equation}
where $\tau_l = |\tau_l \rlangle \tau_l|$ such that
\begin{equation}
        |\tau_l\rangle = \frac{1}{\sqrt{2}}\left( |\eta_{01}\rangle - i(-1)^l|\eta_{10}\rangle\right).
\end{equation}

By the symmetry of this case with that of Option 1b, the mutual information shared between eavesdroppers and Alice is given by
\begin{equation}
    H_{\pm}(\mathcal{Z}:\mathcal{A}) = 1 - h\left[ \big(2 + \sqrt{1- \cos y \cos y'}\big)/4 \right].
\end{equation}
However, when relating this to the detection probability, this case results in a lower value than in the $F=F'=1$ case:
\begin{equation}
    H_{\pm}(\mathcal{Z}:\mathcal{A}) = 1 - h\left[ \big(1 + \sqrt{2d/3}\big)/2 \right]\!.
\end{equation}
This is evident when plotted, as is done in FIG.\,\ref{fig:MI}a.

We end this subsection by noting the mutual information shared between Alice and eavesdroppers, when $F=F'=1$ and we average over the two basis cases:
\begin{multline}\label{ave_MI_EY}
    H(\mathcal{Z} : \mathcal{A}) = 1 - \frac12h\Big[\big(1 + 4\sqrt{d[1-4d]}\big)/2 \Big]\\ - \frac12 h\left[ \big(1 + \sqrt{2d}\big)/2 \right]\!.
\end{multline}
This is plotted in FIG.\,\ref{fig:MI_average}.

\subsection{Alice - Bob mutual information}

The calculation of the mutual information between Alice and Bob, $H(\mathcal{A}: \mathcal{B})$, is less involved than that of $H(\mathcal{Z} : \mathcal{A})$ since Alice and Bob's measurements are fixed. The key thing to note is that the probabilities required for the calculations are found using
\begin{equation}
    P(a_l, b_m) = \Tr\big( \rho_{\text{sifted}} \big|_{k = l, j = m} \big),
\end{equation}
where, similarly to before, $\rho_{\text{sifted}} \big|_{k = l, j = m}$ is made up from the ${k = l, j = m}$ terms in Eq.\,(\ref{rhoKeepUnitary}).

Carrying out all the algebra, we find that, when Alice and Bob measure in the $z$-basis, the mutual information between them is
\begin{equation}
    H_{0/1} (\mathcal{A}: \mathcal{B}) = 1.
\end{equation}
When Alice and Bob measure in the $x$-basis however,
\begin{equation}\label{AB_pm_mututal_info}
    H_{\pm} (\mathcal{A}: \mathcal{B}) = 1 - h\big[ (1+\cos x)/2\big].
\end{equation}
Intuitively, this can be understood when we realise that any errors between Alice and Bob's keys are induced purely by Eve's intervention and not Yves's (this is discussed slightly more in the following subsection). We may therefore expect to see a $\cos x$ dependence but not a $\cos x'$ one.

The form of Eq.\,(\ref{AB_pm_mututal_info}) means we cannot directly determine $H_{\pm} (\mathcal{A}: \mathcal{B})$ [and therefore $H(\mathcal{A}: \mathcal{B})$] using the probability of detection $d$. However, we can use $d$ to put some bounds on $H_{\pm} (\mathcal{A}: \mathcal{B})$. To do this, we look at how Eve and Yves maximise $H (\mathcal{Z}: \mathcal{A})$ for any given value of $d$. Plotting $H_{0/1} (\mathcal{Z}: \mathcal{A})$ with respect to $x,x'$, it can be seen that this maximisation occurs along either the $x = 0$ or $x' = 0$ axis. This results in
\begin{subequations}
    \begin{alignat}{4}
        H_{\pm} (\mathcal{A}: \mathcal{B})|_{x=0} &= 1,\\
        H_{\pm} (\mathcal{A}: \mathcal{B})|_{x'=0} &= 1 - h(4d),
    \end{alignat}
\end{subequations}
where, as before, $0\leq d \leq 1/8$. Thus, if Eve and Yves are aiming to minimise the probability of being detected whilst maximising their knowledge of the shared key,
\begin{subequations}
    \begin{alignat}{4}
        H_{0/1} (\mathcal{A}: \mathcal{B}) &= 1,\\
        H_{\pm} (\mathcal{A}: \mathcal{B}) &\in [1 - h(4d), 1].
    \end{alignat}
\end{subequations}
It is interesting to note that, since $H_{l/\bar{l}}(\mathcal{Z}: \mathcal{A})|_{x = 0}$ and $H_{l/\bar{l}}(\mathcal{Z}: \mathcal{A})|_{x' = 0}$ take the same range of values [as can be seen from Eq.\,(\ref{ZA_MI}) and Eq.\,(\ref{eaves_MI_x})], the eavesdroppers can choose whether or not they want to induce errors in Alice and Bob's shared key whilst extracting information about it. This effectively corresponds to how much impact they allow Eve to have (regardless of Yves). 

Averaging over the different basis choices gives our final range for the mutual information shared between Alice and Bob over the entire protocol:
\begin{equation}\label{ave_MI_AB}
    H (\mathcal{A}: \mathcal{B}) \in \Big[1 - \frac12 h(4d), 1\Big].
\end{equation}
In FIG.\,\ref{fig:MI_average}, we plot these various different mutual information functions, along with $H(\mathcal{A}: \mathcal{Z})$ from Eq.\,(\ref{ave_MI_EY}), with respect to the minimum detection probability $d$ when averaged over the different basis choices of Alice and Bob. Note that in the worst case, Alice and Bob share more information than Alice and the eavesdroppers until a detection probability of $d \approx 0.096$ is induced. This means that up until this point the normal post processing protocols can be undertaken to obtain a secure key between Alice and Bob, at least in the class of attacks considered here \cite{QKDreview}.

We also plot, in FIG.\,\ref{fig:MI}, the various mutual information values for each basis choice. Notice the difference between the $x$ and $z$-basis cases in these plots. This is purely an artifact of how the eavesdroppers' unitaries, and therefore measurements, were set up in a basis dependent way. By an appropriate redefinition of these measurements, we could flip these results and have the eavesdroppers learn something more about the $x$ cases and less about the $z$ ones.
%Indeed, error correction and verification can be performed without publicly declaring any of the key by encrypting such information, either with preshared secret bits, or hash functions \cite{MaQKD}.

\begin{figure*}
        \centering
    \includegraphics[width=
    0.6\textwidth]{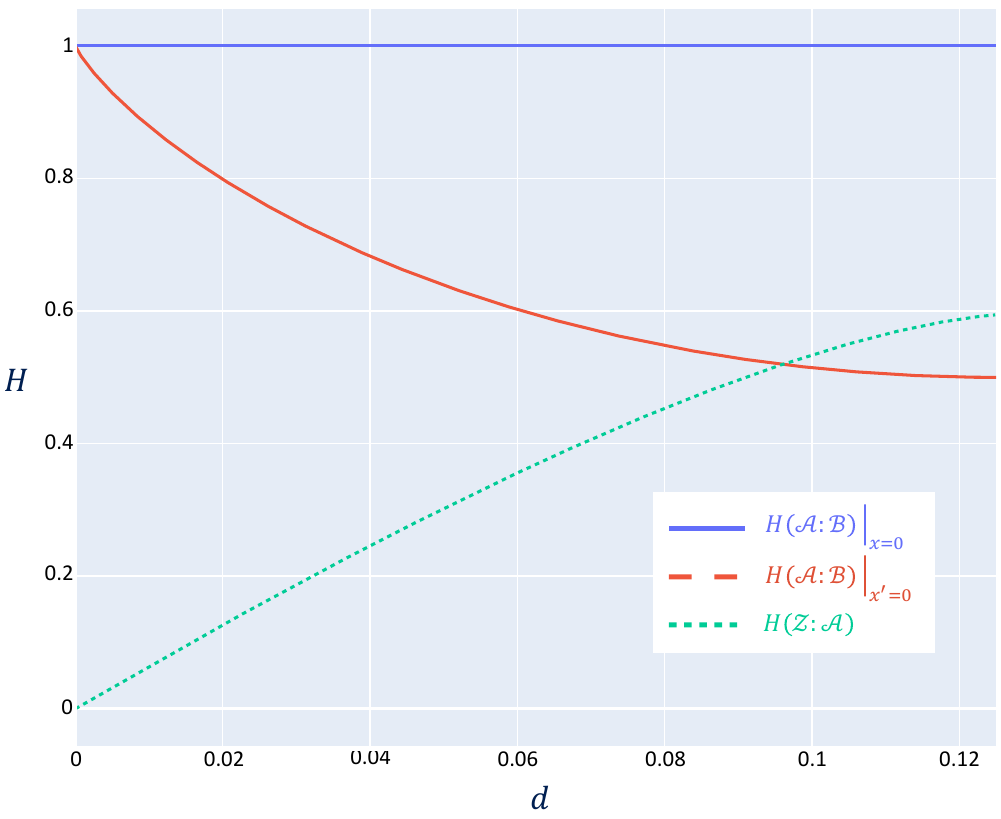}
    \caption{When eavesdroppers Eve and Yves carry out individual (but not fully correlated) attacks, the mutual information (averaged over Alice's basis choices), $H(\mathcal{Z}:\mathcal{A})$, they share with Alice compares to the mutual information between Alice and Bob $H(\mathcal{A}:\mathcal{B})$ [Eq.\,(\ref{ave_MI_AB})] as plotted. These quantities are plotted with respect to the minimum detection probability $d$. In the worst case ($x'$ = 0), Alice and Bob share more information until the detection probability exceeds $d \approx 0.096$.}
    \label{fig:MI_average}
\end{figure*}

\begin{figure*}
        \centering
    \includegraphics[width=
    \textwidth]{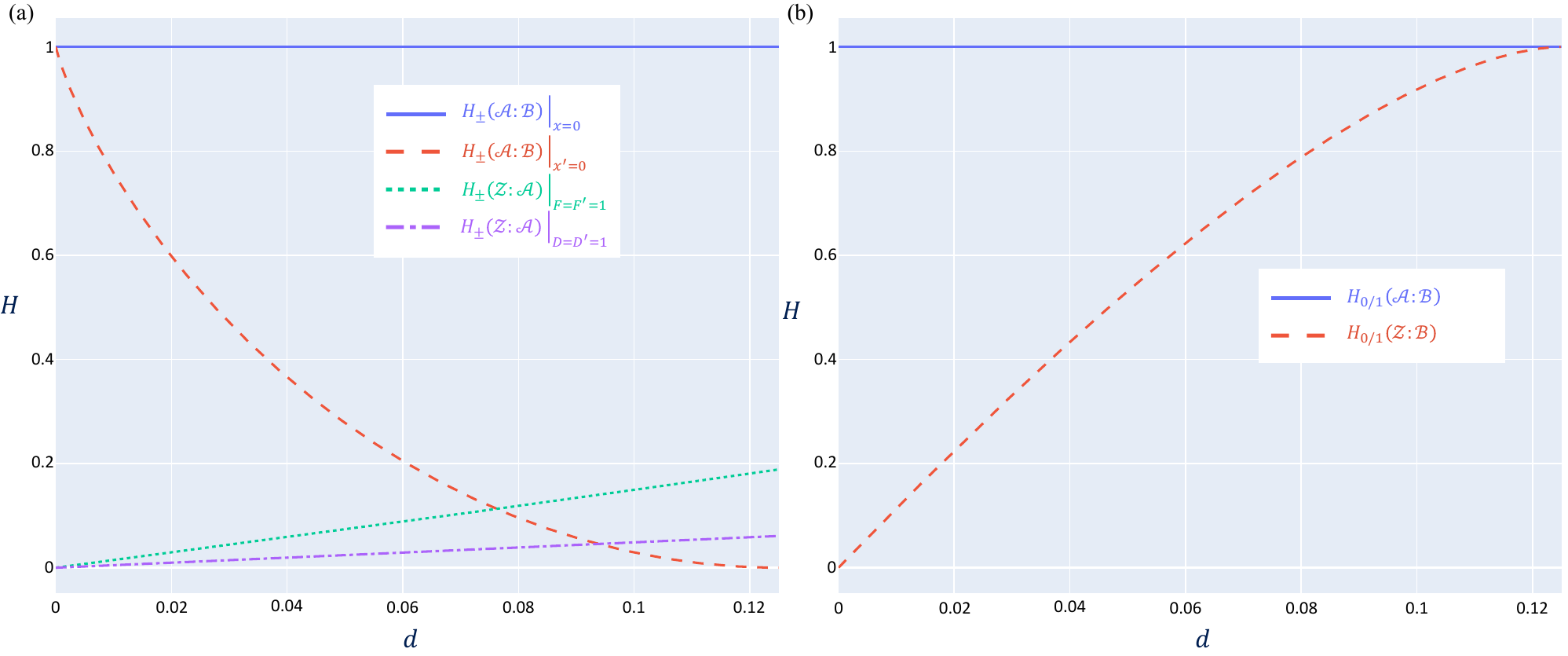}
    \caption{When eavesdroppers Eve and Yves carry out individual (but not fully correlated) attacks, the mutual information between eavesdroppers and Alice $H(\mathcal{Z}:\mathcal{A})$, and the mutual information between Alice and Bob $H(\mathcal{A}:\mathcal{B})$ can differ between Alice and Bob's choice of basis. This is shown by plotting these quantities with respect to the minimum detection probability $d$. In (a), this is shown for the case in which Alice and Bob measure in the $x$-basis, and in (b), Alice and Bob measure in the $z$-basis. Note in (a) the eavesdroppers' mutual information with Alice is always greater in the $F=F'=1$ case than in the $D=D'=1$ case. Indeed, it exceeds Alice and Bob's lower bound at $d \approx 0.076$.}
    \label{fig:MI}
\end{figure*}

\subsection{Error rates}

Let us consider a simple example to gain some intuition as to how the error rate in this protocol differs from its definite causal counterpart. Let us write down the probability of error $P_{\text{error}}$ between Alice and Bob's keys in the case of $F,F' = 1$. Using $P_{\text{error}} = \Tr\big( \rho_{\text{sifted}}\big|_{a \neq b}\big)$, it can be shown that
\begin{equation}\label{ProbError}
    P_{\text{error}} = \frac14(1 - \cos x).
\end{equation}
Note once again that the errors are caused purely by Eve and not Yves. Similarly to in the previous subsection, we can use the maximisation of $H(\mathcal{Z}:\mathcal{A})$ to put bounds on $P_{\text{error}}$. Recall that to do this, we let either $x = 0$ or $x' = 0$, from which it follows that
\begin{equation}
    P_{\text{error}} \in [0, 2d].
\end{equation}

So, let's consider the case in which Eve performs the measurement $\{ |0\rangle\langle0|, |1\rangle\langle 1|\}$. This corresponds to $x = \pi/2$. Using Eq.\,(\ref{SecurityMinDetection}) and Eq.\,(\ref{ProbError}), we can see that $P_{\text{error}} = 1/4$, but $d = 1/8$. This differs from the analogous case in BB84 where the error rate of $1/4$ is used as the detection probability. This is because the effects of the eavesdroppers are not solely contained in the terms used to calculate $P_{\text{detect}}$ (that is, the $\sigma_z \omega \sigma_z$ terms). Some of them lie in the $\omega, \omega \sigma_z$ and $\sigma_z \omega$ terms. It is possible that a different measurement on the control qubit could result in better odds. However, we do not attempt to optimise this measurement here. 

It is interesting to note that the location of the eavesdroppers (performing $\{ |0\rlangle 0|, |1\rlangle1| \}$) dictates the errors induced between Alice and Bob's key. For example, if only Yves is present, there will be no errors found. If, however, Eve is present, regardless of whether Yves is there or not, the probability of error (in this example) between Alice and Bob's key is 1/4, similarly to what is observed in BB84. This phenomenon is true beyond this example: if only Eve is present, every detection event implies an error between Alice and Bob's corresponding key bit, since $[A_k, U_E, A_k] = 0 ~\forall k$. If, however, only Yves is present, there are never any errors induced, since $[U_Y, A_k, A_{\bar{k}}] = 0 ~\forall k$. This is a contrasting point to BB84, where errors are required to detect eavesdroppers. Here, they need not introduce errors to be detected. All in all, this means that, if the location of the eavesdropper(s) is (are) unknown, whether or not errors occur is unknown. These ideas were hinted at in the previous subsection where we saw that the eavesdroppers had the ability to affect how much mutual information Alice and Bob shared.

These ideas allow us to write down an upper bound to the error rate more generally. As shown in Appendix\,\ref{sec:ErrorRateAppendix}, if only Eve is present and induces a detection probability $P_{\text{detect}}$, the error rate is given by $P_{\text{error}}^{E} = 2P_{\text{detect}}$ (where the superscript highlights that only Eve is present). Since Yves does not introduce errors but induces a non-zero detection probability, this provides us with an upper bound to the error rate for any individual attack:
\begin{equation}
    P_{\text{error}} \leq 2P_{\text{detect}}.
\end{equation}

Under the assumption that eavesdroppers are limited to the individual attacks described in this section, recall that the mutual information between the eavesdroppers and Alice is less than or equal to the mutual information shared by Alice and Bob (as shown in FIG.\,\ref{fig:MI}) up until a threshold detection probability $d_{\text{thr}} \approx 0.096$. This means that this protocol is robust to noise that induces an error rate of $P_{\text{error}} \leq 2d_{\text{thr}} \approx 0.192$. This error rate can then be used to inform the normal post processing protocols of error correction and privacy amplification used to obtain a secure key between Alice and Bob \cite{MaQKD}. Of course, the attacks considered here are very limited, so further work is required to understand the information available to Eve and Yves when performing more general attacks, and therefore the true robustness of this protocol to noise. Although this analysis is beyond the scope of this article, we outline how one might go about it in the following section.

\section{Correlated individual attacks and beyond}\label{sec:correlated_and_beyond}

\begin{comment}
As noted before, the strategies considered until now have not allowed Eve and Yves's operations to be correlated prior to the measurement of their ancillae. Allowing for such prior correlations complicates the approach taken here. This is largely due to the increased flexibility afforded to Eve and Yves' unitaries $U_E, U_E$ when we allow their shared ancilla state to be of the form
\begin{equation}
|\Psi\rangle^{\tilde{E}\tilde{Y}} = \sum_{i}\lambda_i |\varepsilon_i\rangle^{\tilde{E}} |\eta_i\rangle^{\tilde{Y}}.
\end{equation}
For instance, to ensure unitarity of $U_E$, instead of requiring expressions such as $F+D = 1$ [Eq.\,(\ref{unitary_conditions}a)], we have ones like $(F_i + D_i) = 1$, which allow for more parameters to be optimised over. Intuitively, we would expect this approach to increase the amount of information they have access to. This intuition comes from the fact that, for uncorrelated attacks, it can be shown that Eve and Yves can learn nothing about the key without introducing a probability of detection. However, as we discuss next, when we allow for correlations between their operations, although they still learn nothing useful about the key, they can learn things such as the basis choice of Alice and Bob.
\end{comment}

As noted before, the strategies considered until now have not allowed Eve and Yves's operations to be correlated prior to the measurement of their ancillae. It turns out that the presence of eavesdroppers, performing individual attacks\footnote{Recall that by individual attacks, we mean that Eve and Yves act on each distributed state separately.}, can be privately detected regardless of the correlations they share. Assuming that the eavesdroppers do not have access to Alice's lab, and therefore cannot alter the causal structure of the protocol by tampering with the control qubit, we let Eve and Yves do anything permitted by quantum mechanics within their respective labs. We therefore allow them to perform correlated {\it quantum instruments}\footnote{Quantum instruments are defined in Appendix \ref{sec:instruments_etc}.}. To describe this situation we use the \textit{process matrix formalism} \cite{Caslav_Orig_Process} introduced in Appendix \ref{sec:instruments_etc}. Indeed, to account for correlations shared between Eve and Yves, we allow them to act jointly on the state being shared via the quantum switch process matrix, together with their own, separate ancillary process matrix. This process matrix is taken to be arbitrary, meaning it describes any possible set of correlations: Eve and Yves could share entangled states; send quantum and classical information to each other; or even utilise their own indefinite causal structure. 
\begin{figure*}
    \centering
    \includegraphics[width=0.99\textwidth]{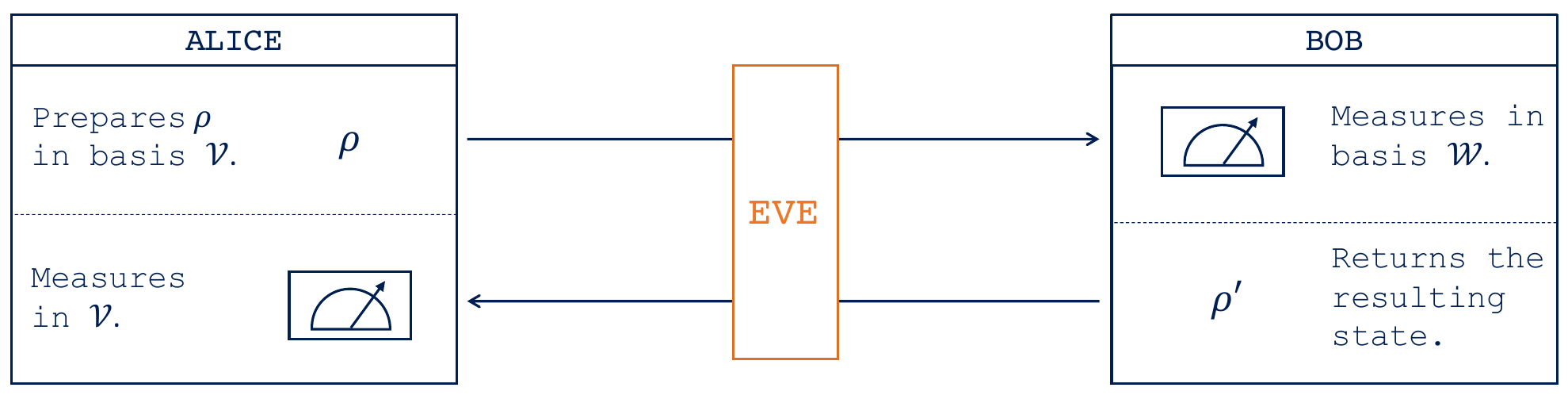}
    \caption{Two way QKD protocol in a definite causal order that realises private detection. The bases $\mathcal{V},\mathcal{W}$ are independently and randomly chosen between the $x$ and $z$-bases.}
    \label{fig:DCOProtocol}
\end{figure*}

Full details of all of this can be found in Appendix \ref{sec:Correlated}, but for clarity, we state our main result here.
\begin{widetext}
\begin{result}
    Define the probability of detection $P_{\text{detect}} = P(\,\omega = |\!\!-\!\rlangle -|\,)$ as the probability of measuring the control qubit to be in the state $|\!-\!\rlangle -|$, and assume eavesdroppers perform arbitrary individual attacks without access to $\omega$. Writing Alice and Bob's measurement outcomes as $i,j$ respectively, and Eve and Yves's joint measurement outcome as $\mathbf{x}$, if $P_{\text{detect}} = 0$, then
    \begin{equation}
    P({\mathbf{x}},i,j) = \begin{cases} \frac{|\langle i|\psi\rangle|^2}{2} \!\left( |r_{\mathbf{x}}^0 + r_{\mathbf{x}}^3|^2 \delta_{ij} + |r_{\mathbf{x}}^1 + r_{\mathbf{x}}^2|^2 \delta_{i\bar{j}} \right)\!, ~~ i,j\in \{0,1\}\\
    \frac{|\langle i|\psi\rangle|^2}{2} \!\left( |r_{\mathbf{x}}^0 + r_{\mathbf{x}}^1|^2 \delta_{ij} + |r_{\mathbf{x}}^2 + r_{\mathbf{x}}^3|^2 \delta_{i\bar{j}}\right)\!, ~~ i,j\in \{+,-\},\end{cases}
\end{equation}
where $|\psi\rlangle \psi| = \rho$ is the initial state of the target system $S$, $r_{\mathbf{x}}^{0},r_{\mathbf{x}}^{1},r_{\mathbf{x}}^{2},r_{\mathbf{x}}^{3} \in \mathbb{C}$, and $\bar{j}$ means ``not $j$".
\end{result}
\end{widetext}
Notice that when undetected (i.e. $P_{\text{detect}} = 0$), Eve and Yves can therefore learn something about Alice and Bob's copies of the keys, but nothing \textit{useful}. They can learn something about which basis was chosen - which is already publicly known - as well as whether errors have occurred between Alice and Bob's keys. Notice also that, similarly to in the previous subsection, Eve and Yves can choose to induce errors in Alice and Bob's key without being detected. But the fact remains that the information the eavesdroppers have access to has no use when it comes to learning about Alice and Bob's key.

One could extend the security proof of Sec.\,\ref{sec:security} to this correlated case by replacing the separable ancilla state $\varepsilon^{\tilde{E}}\otimes \eta^{\tilde{E}}$ shared by Eve and Yves with some shared entangled state $\Phi^{\tilde{E}\tilde{Y}}$. However, to account for fully correlated individual attacks, one would allow eavesdroppers to share an arbitrary process matrix $W^{\tilde{E}\tilde{Y}}$. Intuitively, we would expect both of these approaches to increase the amount of information the eavesdroppers have access to. This intuition comes from the fact that, for uncorrelated attacks, it can be shown that Eve and Yves can learn nothing about the key without introducing a probability of detection. However, as discussed above, when we allow for correlations between their operations, although they still learn nothing useful about the key, they can learn things such as the basis choice of Alice and Bob.

To extend the eavesdroppers strategy further to collective attacks, suppose Alice and Bob store their measurement results in an ancilliary state by performing their measurements using the Kraus operators described in Sec.\,\ref{sec:NoEve}, $A'_i = |i\rangle\langle i|^S \otimes |m_i\rangle \langle m_0|^{\tilde{A}}/\sqrt{2}, ~ B'_i = |i\rangle\langle i|^S \otimes |m_i\rangle\langle m_0|^{\tilde{B}}/\sqrt{2}$ for $i\in \{0,1,+,-\}$:
\begin{equation}
    \rho^S \otimes |m_0\rlangle m_0|^{\tilde{A}} \to \sum_i A_i \rho^S A_i^{\dagger} \otimes |m_i\rlangle m_i|^{\tilde{A}},
\end{equation}
and similarly for Bob. If eavesdroppers perform the same operation on each of the $n$ runs of the protocol, the final state $(\rho^{\tilde{A}\tilde{B}\tilde{E}\tilde{Y}})^{\otimes n}$ shared within Alice, Bob, Eve and Yves' ancilliary systems $\tilde{A}, \tilde{B}, \tilde{E}$ and $\tilde{Y}$ respectively, will be $n$ i.i.d. copies of the same state. Allowing Eve and Yves to perform a global measurement on the final state $\Phi^{\tilde{E}\tilde{Y}} = \Tr_{\tilde{A}\tilde{B}}\big[(\rho^{\tilde{A}\tilde{B}\tilde{E}\tilde{Y}})^{\otimes n}\big]$ of their system, the asymptotic key rate, secure against these {\it collective attacks}, could be calculated via \cite{murta2023lecture}
\begin{equation}
    r = S(\tilde{A}|\tilde{E}\tilde{Y}) - S(\tilde{A}|\tilde{B}),
\end{equation}
where 
\begin{equation}
    S(X|Y) = S(XY) - S(Y),
\end{equation}
$S(V) = -\Tr \big[ \rho^V \log \rho^V \big]$, and $\rho^Y = \Tr_X (\rho^{XY})$. 

The quantum de Finetti theorem could then be employed to account for coherent attacks \cite{renner2008security, Collective_to_Coherent, MaLM05}. If the final state after $n$-rounds with eavesdroppers performing a correlated coherent attack is given by $\Phi^{\tilde{A}\tilde{B}\tilde{E}\tilde{Y}}$, the quantum de Finetti theorem says that having Alice choose a selection of $n' \ll n$ ancilla systems and randomly permuting them results in the state $\big(\sigma^{\tilde{A}\tilde{B}\tilde{E}\tilde{Y}}\big)^{\otimes n'}$ for some unknown $\sigma^{\tilde{A}\tilde{B}\tilde{E}\tilde{Y}}$. It follows that if this protocol is $\epsilon$-secure against collective attacks at generating a key of length $n$ at a rate $r$, then it is $\epsilon'$-secure (where $\epsilon'>\epsilon$) against coherent attacks at generating a key of length $n'$ at a rate $r$ \cite{Collective_to_Coherent}.

The challenge of carrying this general security proof is in the calculation of the von Neumann entropies required to put bounds on $r$: $-\Tr \big(\sigma^{\tilde{A}\tilde{B}\tilde{E}\tilde{Y}} \log \sigma^{\tilde{A}\tilde{B}\tilde{E}\tilde{Y}}\big),$ $ -\Tr \big(\sigma^{\tilde{A}\tilde{B}} \log \sigma^{\tilde{A}\tilde{B}}\big),$ $ -\Tr \big(\sigma^{\tilde{B}} \log \sigma^{\tilde{B}}\big)$. The allowed correlations between eavesdroppers are more numerous than in the definite causal case, meaning the states $\sigma^{\tilde{A}\tilde{B}\tilde{E}\tilde{Y}}$ are more complex.

\section{Private detection in a definite causal order}\label{sec:DefiniteCausalOrder?}

\noindent In this section, we consider whether this same phenomena of private detection could be achieved {\it without} the help of indefinite causal order. Indeed, we provide here evidence that it {\it is} possible. However an extra measurement by Alice is required, and as we shall discuss, some attacks possible here are not possible in the indefinite causal case we have considered. This indicates a more subtle relationship between definite and indefinite causal quantum key distribution protocols.

Suppose, as depicted in FIG.\,\ref{fig:DCOProtocol}, Alice prepares a state $\rho\in \mathcal{L}(\mathcal{V})$ that is either the $x$ or $z$-basis (with corresponding key bits as before). Suppose she then sends $\rho$ to Bob, who measures in the basis $\mathcal{W}$ which, again, is chosen randomly between the $x$ and $z$-bases. Following this, Bob returns the updated state back to Alice who measures it in the {\it same} basis that she prepared the state in: $\mathcal{V}$. As they did in the other sections of this chapter, Alice and Bob then compare which bases they chose, and only use the cases in which they agree for their shared key. Now, by counting how often she measures a different state from the one she prepared, a scenario that we call an error, Alice can monitor for eavesdroppers. This is possible since she knows {\it two things} when there are no eavesdroppers:
\begin{enumerate}
    \item When $\mathcal{V} = \mathcal{W}$, the probability of error $P_{\text{error}}^{\text{same}}$ is zero.\footnote{For example, suppose Alice prepares the state $|\! -\rlangle -|$ and Bob measures it in the basis $\{|\!\pm \rlangle \pm|\}$. Since $\langle + | - \rangle = 0$, Bob is guaranteed to obtain an outcome of $-$, and therefore send the state $|\!-\rlangle -|$ back to Alice. So, when Alice goes on to measure this state in the basis she prepared it in: $\{ |\!\pm\rlangle \pm|\}$, she will always obtain a result of $-$. In other words, she will never register an error.}
    \item When $\mathcal{V}\neq \mathcal{W}$, the probability of error $P_{\text{error}}^{\text{diff}}$ is $1/2$.\footnote{For example, suppose Alice prepares the state $|0\rlangle 0|$ and Bob measures in the basis: $\{|\!\pm\rlangle \pm| \}$. Then the state Alice receives from Bob is $\mathbbm{1}/2$, meaning that when Alice measures this in her original basis $\{ |0/1\rlangle 0/1|\}$, the probability she measures the state to be different to the one she prepared is $\langle 1| \mathbbm{1} | 1\rangle/2 = 1/2$.}
\end{enumerate}
Let's consider an example showing how an eavesdropper's intervention affects at least one of these probabilities. 

Suppose an eavesdropper, Eve, attempts to keep $P_{\text{error}}^{\text{same}}$ at its expected value of zero. As depicted in FIG.\,\ref{fig:DCO_attack}, she can do this by sending out a probe state $\xi$ to Bob, whilst returning Alice's qubit state $\rho$ back to her, unaffected\footnote{In FIG.\,\ref{fig:DCO_attack} this would correspond to setting the operation $\mathcal{E}$ equal to the identity $\mathcal{I}$.}. In this scheme, Eve can access as much information as possible about Bob's key without affecting $P_{\text{error}}^{\text{same}}$. However, this strategy would also result in $P_{\text{error}}^{\text{diff}} = 0 \neq 1/2$, thereby allowing Alice to detect Eve. This is, of course, an extreme case: Eve could attempt to find out which measurement Bob performed and induce errors using some operation $\mathcal{E}$ on $\rho$ to increase $P_{\text{error}}^{\text{diff}}$. But it doesn't seem as though she could ever increase it to 1/2. 

To see this, suppose Eve thinks she knows which measurement Bob chose, for example, $\{ |\pm \rlangle \pm|\}$. Then she can introduce errors in at most 50\% of the bits by acting on $\rho$ with $\sigma_x$: if Alice prepared $|\pm\rangle$, no errors are induced (as Alice would expect if Bob measured in the $x$-basis), if Alice prepared in $|0/1\rangle$, Eve induces an error. The problem is, Bob is choosing randomly between mutually unbiased bases, meaning Eve can only be sure which basis he chose some proportion $p < 1$ of the time\footnote{As is done in Appendix \ref{sec:two-way_attack}, unambiguous quantum state discrimination can be used to analyse this tactic more concretely \cite{SarahQSD,bergouQSD,Bergou_Unambigous_Mixed}.}. It follows that, if Eve wants to preserve $P_{\text{error}}^{\text{same}} = 0$, the attack proposed in FIG.\,\ref{fig:DCO_attack} results in, at most, $P_{\text{error}}^{\text{diff}} = p/2 < 1/2$. We discuss this situation more explicitly in Appendix \ref{sec:two-way_attack}. Of course, more general correlations between Eve's probe state $\sigma$, and how she acts on Alice's state $\rho$ are possible. That being said, we found no strategy that allowed for Eve to extract information whilst leaving both $P_{\text{error}}^{\text{same}}$ and $P_{\text{error}}^{\text{diff}}$ unaltered. We leave the full analysis of this scenario for future work.

So, it may seem as though there is no difference between the indefinite and definite causal cases. However, one thing to notice here is that the attacks discussed here (depicted in FIG.\,\ref{fig:DCO_attack}) are not ones that are possible in our indefinite causal protocol. In order for undetected eavesdroppers to send $\rho$ back to Alice without going via Bob, Eve and Yves would need to do so coherently in the superposition induced by $\omega$. Doing so requires them to have access to the control qubit $\omega$. We discuss possible implications of this realisation in the following section.

\begin{figure*}
    \centering
    \includegraphics[width=0.99\textwidth]{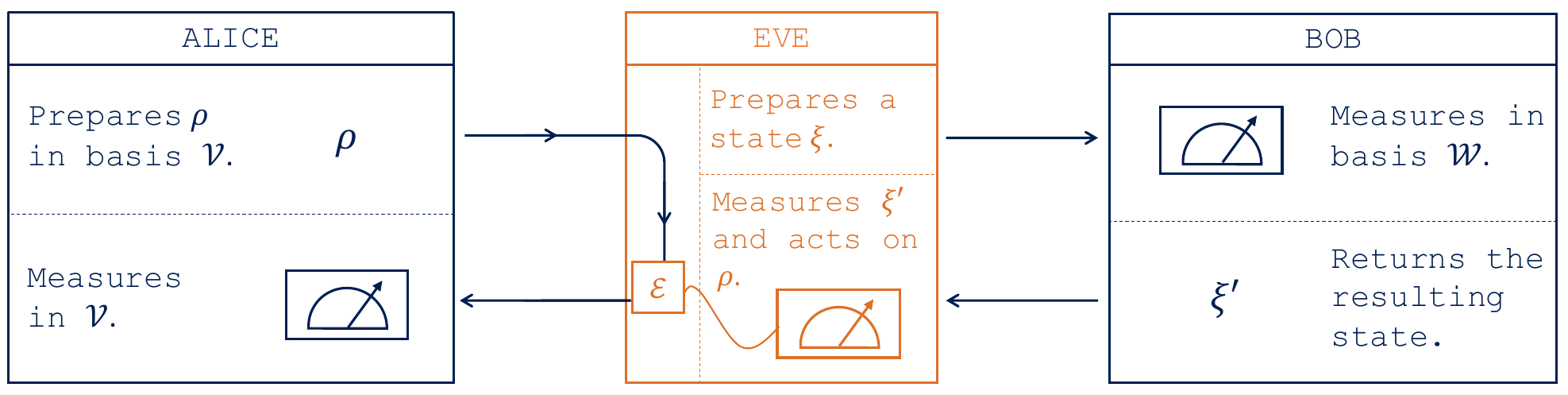}
    \caption{Eavesdropping strategy allowing for Alice's state to be returned to her unaffected whilst letting Eve learn something about Bob's key. If the probability $P_{\text{error}}^{\text{same}}$ is left unaltered at its expected value of zero, we found no strategy that allowed the probability $P_{\text{error}}^{\text{diff}}$ to remain at its expected value of $1/2$.}
    \label{fig:DCO_attack}
\end{figure*}

\section{Discussion and Conclusion}\label{sec:Conc}

In this work we explored the idea of performing the QKD protocol BB84 in an indefinite causal regime. We defined a protocol that achieves this by performing projective measurements in an indefinite causal order. In doing so, we found that it is possible to detect eavesdroppers during a QKD task {\it without} publicly comparing any subset of a shared private key between the two parties involved, Alice and Bob. We found that this could be achieved using a second system that acts as the control in inducing the indefinite causal ordering. In contrast to one-way QKD protocols, but similarly to two-way protocols, there are two locations eavesdroppers can reside, allowing for cooperative attacks. These have both been considered and the security against a class of individual attacks by the eavesdroppers was proven. Further, it was shown (in Appendix\,\ref{sec:Correlated}) that eavesdroppers, who do not have access to the qubit controlling the causal order of the protocol, cannot extract any useful information about the key without inducing a non-zero detection probability, at least when they act on each distributed state individually. Indeed, we showed that this was true for {\it any} cooperative strategy the eavesdroppers used. In Sec.\,\ref{sec:DefiniteCausalOrder?}, we discussed a possible way of privately detecting eavesdroppers using a two-way protocol in a definite causal order. To do this, an extra instance of Alice's operation was required, a property consistent with other discussions of indefinite versus definite causal orderings \cite{Quantum_Switch_original}. 

This ICO QKD protocol could be challenging to realise experimentally due to the requirement to preserve coherence between the two causal orderings over the distance of communication. Further, one might expect projective measurements of $\rho$, performed in an indefinite causal order, to be unattainable. However, as discussed in Appendix \ref{sec:Exp}, the results of this protocol could be simulated using linearly polarised light, a Sagnac interferometer and some polarising filters. The Sagnac interferometer would create the indefinite causal order and the polarising filters would be orientated in various different ways to correspond to each of Alice and Bob's measurement outcomes. More promising is with regards to recent work by H. Cao et al \cite{InstrumentsInICO}, who experimentally performed quantum instruments in an indefinite causal order. This indicates that a full experimental implementation of this protocol could be possible in the not-too-distant future.

When it comes to practicality, consider using a Sagnac interferometer or something similar to create an indefinite causal ordering of operations. In order for the ICO to be legitimate, the coherence length of the light used must be considerably larger than the path length of the interferometer \cite{QuantumSwitch_Romero}, perhaps indicating a limit to how practical such a protocol would be. Another limitation becomes apparent when we notice that two qubits are required to distribute one key bit securely, compared to BB84's one qubit. Perhaps this second, control state could find a secondary use beyond determining the presence of eavesdroppers, but we leave this consideration for future work. We also note that the effects of noise and loss in this protocol have not been considered here. Indeed, being similar in nature to two-way protocols, they are likely to have a substantial impact (in comparison to one-way protocols). Having said this, there is evidence that the effects of noise can be evaded in certain indefinite causal scenarios \cite{Indef_Noise,Channel_ICO,Channel_ICO1, Aaron_ICO} which could make this an interesting line of future research\footnote{Regardless of whether this phenomenon is a unique to indefinite causal order or not \cite{abbott}.}. Finally, it should be noted that the main difference between this approach and previously studied QKD protocols is that a subset of the sifted key need not be publicly compared and subsequently discarded to detect Eve. Since this subset is usually small compared to the length of the sifted key, the cost of performing QKD in an ICO certainly outweighs the benefit of private detection.

Having said this, for reasons stemming from this article, we argue that further exploration into the crossover between quantum cryptography and indefinite causal structures is warranted. Recall that in Sec.\,\ref{sec:DefiniteCausalOrder?} we discussed a definite causal protocol that exhibited similar phenomena to our indefinite casual one. Here, we considered various eavesdropping attacks, in particular the one depicted in FIG.\,\ref{fig:DCO_attack}. We noted that these attacks are not possible in our indefinite causal protocol: undetected eavesdroppers can't send $\rho$ straight back to Alice without having access to the control qubit state $\omega$. This hints at an intriguing thought: do eavesdroppers have to be able to ``hack into" and alter the causal structure of a protocol to realise their full, cryptanalytic potential? Indeed, if so, could this mean that the security of a key sharing protocol is bolstered by obscuring its causal structure? It's plausible that eavesdroppers may be able to somehow gain access to the control qubit in the quantum switch scenario we considered, but how about when different, more obscure, causal structures are used? These more abstract causal structures often require more than two-parties, so could they be employed in different cryptographic scenarios, such as conference key agreement \cite{murtaConferenceKey, RefRef4, RefRef5}? Related to these questions, how do the ideas of device-independent QKD \cite{colbeck_Device_independence,measurement_Device_Independence, Fully_device_independent} translate into a setting where eavesdroppers may be infiltrating the causal structure of the protocol? 

More practically, as mentioned before, the possibility of reducing the effects of noise \cite{Indef_Noise,Channel_ICO,Channel_ICO1, Aaron_ICO} could be another interesting avenue to explore. Could this feature of ICO lead to improved key rates, particularly within two-way protocols where two copies of a noisy channel are naturally present \cite{LM05, beaudry2013security}? On the flip side, by utilising this noise reducing phenomenon, could eavesdroppers somehow use indefinite causal structures to reduce their introduction of errors? This could impact topics involving sequential quantum measurement more broadly, whether it be in quantum state discrimination \cite{Sequential1} or quantum learning \cite{OurPaper}. We therefore consider this work as just an initial step in studying quantum cryptographic protocols within indefinite causal structures.

%Having said all this, although two qubits are indeed used, it is still true that only one need be sent between Alice and Bob: the second can be kept in Alice's lab.

\begin{acknowledgements}
The author would like to thank Sarah Croke and John Jeffers for the invaluable discussions. The author acknowledges The Engineering and Physical Sciences Research Council and the UK National Quantum Technologies Programme via the QuantIC Quantum Imaging Hub (EP/T00097X/1). This work benefitted from network activities through the INAQT network, supported by the Engineering and Physical Sciences Research Council (grant number EP/W026910/1). This research was supported in part by Perimeter Institute for Theoretical Physics. Research at Perimeter Institute is supported by the Government of Canada through the Department of Innovation, Science and Economic Development and by the Province of Ontario through the Ministry of Colleges and Universities.
\end{acknowledgements}

%%%%%%%%%%%%%%%%%%%%%%%%%
%%%%%%%%%%%%%%%%%%%%%%%%%

\appendix

\section{A definite causal QKD attack}\label{sec:two-way_attack}

In Sec.\,\ref{sec:DefiniteCausalOrder?}, we considered an attack, depicted in FIG.\,\ref{fig:DCO_attack} and discussed how $P_{\text{error}}^{\text{same}} =0, P_{\text{error}}^{\text{diff}} =1/2,$ cannot seem to be simultaneously achieved. We expand on this argument in this short appendix.

Suppose Eve sends the probe state $\xi$ to Bob. Depending on Bob's measurement choice, he would send back one of the following two states:
\begin{subequations}\label{xi_updates}
    \begin{alignat}{4}
        \xi_{z} &= \langle 0 | \xi | 0 \rangle |0\rlangle 0| + \langle 1|\xi | 1\rangle |1\rlangle 1|,\\
        \xi_{x} &= \langle + | \xi | + \rangle |\!+\rlangle +| + \langle -|\xi | -\rangle |\!-\rlangle -|.
    \end{alignat}
\end{subequations}
Eve's aim is to use these returned states to introduce or avoid errors in the relevant places. Since she doesn't want to induce any errors when Alice and Bob {\it agree} on their basis choice, only when they \textit{disagree}, she must be sure which basis Bob chose. She could therefore perform unambiguous quantum state discrimination \cite{SarahQSD,bergouQSD,Bergou_Unambigous_Mixed} to distinguish $\{\xi_{z}, \xi_{x} \}$. To do this, she uses the following POVM
\begin{subequations}
    \begin{alignat}{4}
        \pi_{z} &= \alpha\xi_{x}^{\perp},\\
        \pi_{x} &= \beta \xi_{z}^{\perp},\\
        \pi_{?} &= \mathbbm{1} - \pi_x - \pi_z,
    \end{alignat}
\end{subequations}
where $\alpha, \beta \geq 0$, $\Tr(\xi_{w}^{\perp} \xi_{w}) = 0$ for $w = x,z$, and $\pi_?$ corresponds to the inconclusive outcome. 

%Now, $\xi_{z}^{\perp} \neq 0 ~(\xi_{x}^{\perp} \neq 0)$  only if $\text{ker}(\xi_{z})\neq \varnothing ~(\text{ker}(\xi_x) \neq \varnothing)$. Indeed, since the $z$ and $x$ measurements are mutually unbiased, either $\text{ker}(\xi_z) \neq \varnothing$ {\it or} $\text{ker}(\xi_x) \neq \varnothing$, but never both. Looking back at Eq.\,(\ref{xi_updates}), it follows that $\exists!\, i\in \{0,1,+,-\}$ such that $\langle i | \xi | i \rangle = 0$.

Now, since $\pi_x = \pi_z = 0$ would make for an uninformative measurement, we assume, without loss of generality, that $\xi_z^{\perp} \neq 0$. It follows that the support of $\xi_z$ is not $\mathbb{C}^2$, and so, from Eq.\,(\ref{xi_updates}), $\exists!\, i\in \{0,1\}$ such that $\langle i | \xi | i \rangle = 0$. Again without loss of generality, let $\langle 1 | \xi | 1 \rangle = 0$ and thus $\xi = |0\rlangle 0 |$. This means that $\xi_z = |0\rlangle 0|, \xi_x = \mathbbm1/2$, and therefore 
\begin{subequations}
    \begin{alignat}{4}
        \pi_z &= 0,\\
        \pi_x &= \beta |1\rlangle 1|,\\
        \pi_? &= |0\rlangle 0| + (1-\beta)|1\rlangle 1|.
    \end{alignat}
\end{subequations}
Minimising the probability of an inconclusive outcome results in $\beta = 1$, meaning the probability that Eve is sure Bob measured in the $x$-basis is $ p = 1/4 < 1$ (assuming Bob picks the $x$-measurement 1/2 of the time). So, by acting on the state $\rho$, being sent back to Alice, with $\sigma_x$ whenever she measures $\pi_x$, Eve induces an error $P_{\text{error}}^{\text{diff}} = 1/8 < 1/2$.

\section{Quantum operations, instruments, the Choi-Jamio{\l}kowski isomorphism and process matrices}\label{sec:instruments_etc}

\subsection{Quantum operations and instruments}

We begin this appendix by defining a {\it quantum operation}. In what follows we take $\mathcal{L}(\mathcal{V})$ to be the space of linear operators on $\mathcal{V}$, and we will label out quantum system by $S$. 
\begin{definition}
    A quantum operation is a map $\mathcal{E}:\mathcal{L}\big(\mathcal{H}^S\big) \to \mathcal{L}\big( \mathcal{H}^{S'} \big)$ satisfying the following three conditions:
    \begin{enumerate}
    \item For any density operator $\rho \in \mathcal{L}\big(\mathcal{H}^S\big)$, $\Tr \mathcal{E}(\rho) \in [0,1]$.
    \item It is {\it convex-linear}. That is, for any convex combination of density operators $\sum_i p_i \rho_i$, \begin{equation*}
        \mathcal{E} \Big( \sum \limits_i p_i \rho_i \Big) = \sum \limits_i p_i \mathcal{E}(\rho_i).
    \end{equation*}
    \item It is completely-positive (CP).
    \end{enumerate}
\end{definition}
\noindent One particularly important class of quantum operations are called {\it quantum channels}. These are subject to the additional constraint that they are {\it trace-preserving}. That is, $\mathcal{E}_c:\mathcal{L}\big(\mathcal{H}^S\big) \to \mathcal{L}\big( \mathcal{H}^{S'} \big)$ is a quantum channel if it is a quantum operation and, for any density operator $\rho \in \mathcal{L}\big(\mathcal{H}^S\big)$, $\Tr \mathcal{E}_c(\rho) = 1$. Therefore, quantum channels are often called {\it completely-positive trace-preserving (CPTP)} maps.

An alternative, but equivalent, way of describing a quantum operation $\mathcal{E}$ mathematically is to use a set of operators $\{ E_i \}$ called Kraus operators, named after K. Kraus who first noted this equivalence \cite{kraus, SarahKraus}. As before, let $\rho$ be a density operator in the input space of $\mathcal{E}$, then there exists a set of operators $\{E_i\}$, subject to
\begin{equation}\label{KrausCondition}
    \sum \limits_i E_i^{\dagger} E_i \leq \mathbb{I},
\end{equation}
that allow us to write
\begin{equation}
    \mathcal{E} (\rho) = \sum \limits_i E_i \rho E_i^{\dagger}.
\end{equation}
Being a sum of positive operators, the condition in Eq.\,(\ref{KrausCondition}) corresponds to the first axiom of quantum operations: $\Tr \mathcal{E}(\rho) \in [0,1]$. If $\mathcal{E}$ is CPTP (i.e. a quantum channel), then this condition becomes an equality.

We are now able to define the mathematical structure we use to describe everything allowed by quantum mechanics within a closed lab: a \textit{quantum instrument} \cite{QuantumInstrument, Caslav_Orig_Process}.
\begin{definition}
    A quantum instrument is a set of quantum operations $\mathcal{M} = \big\{\mathcal{M}_i : \mathcal{L}\big(\mathcal{H}^S\big) \to \mathcal{L}\big(\mathcal{H}^{S'}\big) \big\}$, such that $\sum_i \mathcal{M}_i$ is CPTP.
\end{definition}
\noindent When performing the instrument $\mathcal{M}=\{\mathcal{M}_i\}$, an outcome ``$i$" is registered which tells us that the quantum operation $\mathcal{M}_i$ has occurred. The probability of obtaining an outcome $i$, given that $S$ was prepared in the state $\rho$, is given by the (generalised) Born rule: $P(i|\rho) = \Tr \mathcal{M}_i (\rho)$.

Quantum instruments can rewritten in the Kraus representation. That is, for each outcome $i$ of $\mathcal{M}$, the corresponding operation $\mathcal{M}_i$ can be decomposed as follows:
\begin{equation}
    \mathcal{M}_i (\rho) = \sum \limits_j M_{ij} \rho M_{ij}^{\dagger},
\end{equation}
such that the Kraus operators $\big\{ M_{ij} \big\}$ satisfy
\begin{equation}
    \sum \limits_{i,j} M_{ij}^{\dagger} M_{ij} = \mathbb{I}.
\end{equation}
This corresponds to the requirement that $\sum_i \mathcal{M}_i$ be CPTP in order to ensure that the probabilities $P(i|\rho) = \Tr \mathcal{M}_i (\rho)$ sum to unity. If this wasn't the case, there would be some other possible outcome unaccounted for. Note that it follows that
\begin{equation}
    \sum \limits_i {M_{ij}}^{\dagger} M_{ij} \leq \mathbb{I},
\end{equation}
which corresponds to the definition of a quantum operation, as we'd hope since $\mathcal{M}_i$ is one.

\subsection{The Choi-Jamio{\l}kowski isomorphism and process matrices}\label{sec:Choi_and_process}

In the following appendix, we aim to consider Eve and Yves who can work together with the help of classical, quantum and (indefinite) causal correlations to perform individual attacks. In order to do this, it is convenient to use the \textit{process matrix formalism} \cite{Witnessing_causal_nonseparability}. This formalism essentially only insists that quantum mechanics is obeyed {\it locally}. Indeed, no assumptions are made about the causal structure between parties. 

To utilise this formalism, we employ the \textit{Choi-Jamio{\l}kowsi isomorphism} \cite{dePillis_pre_jamiolkowski,choi,jamiolkowski} to reinterpret quantum instruments as sets of positive semi-definite operators. More specifically, suppose some system $X_I$ with Hilbert space $\mathcal{H}^{X_I}$ is input to some closed lab and is evolved via some CP map $\mathcal{X}: \mathcal{L}(\mathcal{H}^{X_I}) \to \mathcal{L}(\mathcal{H}^{X_O})$ resulting in the system $X_O$, with corresponding Hilbert space $\mathcal{H}^{X_O}$, being output\footnote{For our purposes, this could be the lab of Alice, Bob, Eve or Yves.}. The Choi-Jamio{\l}kowski (CJ) isomorphism details a correspondence between such CP maps $\mathcal{X}$ and positive semi-definite operators $M^X \in \mathcal{L}(\mathcal{H}^{X_I}) \otimes \mathcal{L}(\mathcal{H}^{X_O})$, which we call CJ operators. As we have done here, we will often use the abbreviation $X := X_I X_O$. Explicitly, 
\begin{equation}
M^X = (\mathcal{I} \otimes \mathcal{X}) (|\mathbbm1 \rangle\!\rangle ^{X_I X_I'} \langle\!\langle \mathbbm1|^{X_I X_I'}),
\end{equation} 
where $|\mathbbm1 \rangle\!\rangle^{X_I X_I'} = \sum _{j} |jj\rangle^{X_I X_I'}$ for some \textit{choice} of complete basis $\{|j\rangle\} \subset \mathcal{H}^{X_I}$ and the primed superscript $X_I'$ indicates that the space $\mathcal{H}^{X_I'}$ is a copy of $\mathcal{H}^{X_I}$. For Alice and Bob, we will use $j \in \{0,1\}$. The map $\mathcal{X}$ can be recovered using
\begin{equation}\label{channel_recoved}
\mathcal{X} (\rho) = \Tr_{X_I} \! \left[ (\rho^T \otimes \mathbbm1) M^X \right]
\end{equation}
for some state $\rho$, where the superscript $T$ denotes the transpose with respect to the chosen basis $\{|j\rangle\}$. 

Since quantum instruments are made up of quantum operations (maps which are CP, among other things), we can use the Choi-Jamio{\l}kowski isomorphism to describe quantum instruments as sets of CJ operators. So, suppose two parties, Alice and Bob perform the respective instruments $\mathcal{M}^{A} = \{\mathcal{M}_i^A : \mathcal{L}(\mathcal{H}^{A_I}) \to \mathcal{L}(\mathcal{H}^{A_O})\}, \mathcal{M}^{B} = \{\mathcal{M}_j^B: \mathcal{L}(\mathcal{H}^{B_I}) \to \mathcal{L}(\mathcal{H}^{B_O})\}$ on some system that passes through their labs\footnote{They could be acting on the same system that is shared between them through some quantum channel, they could be acting on their respective parts of an entangled state, and so on - no assumptions are made here.}. The probability that Alice obtains an outcome $i$ and Bob $j$, is given by what we call a process \cite{Caslav_Orig_Process}:
\begin{equation}
    P(i,j) = \Tr \Big[ \big(M^A_i \otimes M^B_j \big)^T W^{AB}\Big],
\end{equation}
where $M^A_i, M^B_j$ are the CJ operators corresponding to $\mathcal{M}_i^A, \mathcal{M}^B_j$ respectively, and $T$ is the transpose with respect to the chosen bases in the definitions of $M^A_i, M^B_j$. Also in this expression is the \textit{process matrix} $W^{AB} ~(\equiv W^{A_IA_OB_IB_O})$ which details all the possible correlations between Alice and Bob's outcomes. It is defined as a linear operator on $\mathcal{H}^{A_I}\otimes \mathcal{H}^{A_O} \otimes \mathcal{H}^{B_I} \otimes \mathcal{H}^{B_O}$ satisfying
\begin{subequations}
    \begin{alignat}{4}
        W^{AB} &\geq 0,\\
        \Tr \Big[ \big(M^A \otimes M^B \big)^T W^{AB}\Big] &= 1
    \end{alignat}
\end{subequations}
for all CJ operators ${M}^A =\sum_i {M}_i^A, ~ {M}^B= \sum_j {M}_j^B$  corresponding to the CPTP maps $\sum_i \mathcal{M}_i^A, ~ \sum_j \mathcal{M}_j^B$ respectively.

\section{Fully correlated eavesdroppers}\label{sec:Correlated}

The aim of this appendix is to see if our protocol withstands eavesdroppers who can work together, as depicted in FIG.\,\ref{fig:YBEA}. In particular, eavesdroppers, performing individual attacks, who cannot alter the indefinite causal structure of the key sharing device. In other words, we do not allow them access to the control qubit $\omega$ of FIG.\,\ref{fig:YBEA}. Let's begin by setting up the problem.

\subsection{Problem setup}

In this setup, we have multiple labs: first, Bob, labelled by $B$ who performs the same channel as always $\mathcal{B}$ with Kraus operators $\{B_j\}$ defined in Eq.\,(\ref{Kraus}). Second, Alice, who we can think of as having two ``sublabs": the one we've discussed before, labeled by $A$ where she performs the channel $\mathcal{A}$ with Kraus operators $\{ A_i \}$, again defined in Eq.\,(\ref{Kraus}). But, following \cite{switch_process_1, Witnessing_causal_nonseparability}, she also has another sublab, labeled by $C$ that takes in the target and control states at the end. That is, we think of this final system, with Hilbert space $\mathcal{H}^{C_t} \otimes \mathcal{H}^{C_c}$, composed from a target component and control component $C_t, C_c$ respectively, as residing within Alice's lab.

Eve and Yves, on the other hand, are only required to obey quantum mechanics locally within their labs, and so can perform quantum instruments $\mathcal{E}, \mathcal{Y}$ respectively. What's more, since they are allowed to work together, not only do they have access to the state being distributed (with the system passing through their labs labeled by $E,Y$ respectively), but they also share some ancillary process matrix $W^{\tilde{E}\tilde{Y}}$ detailing all possible correlations between the outcomes of their quantum instruments. We therefore write their quantum instruments as follows:
\begin{subequations}\label{EveAndYvesInstruments}
    \begin{alignat}{4}
        \mathcal{E} &= \big\{ \mathcal{E}_l: \mathcal{L}(\mathcal{H}^{E_I} \otimes \mathcal{H}^{\tilde{E}_I}) \to \mathcal{L}(\mathcal{H}^{E_O} \otimes \mathcal{H}^{\tilde{E}_O}) \big\}_l,\\
        \mathcal{Y} &= \big\{ \mathcal{Y}_p: \mathcal{L}(\mathcal{H}^{Y_I} \otimes \mathcal{H}^{\tilde{Y}_I}) \to \mathcal{L}(\mathcal{H}^{Y_O} \otimes \mathcal{H}^{\tilde{Y}_O}) \big\}_p,
    \end{alignat}
\end{subequations}
where $E_{I}, Y_{I} ~ (E_{O}, Y_{O})$ are the systems received (sent), by Eve and Yves from (to) Alice and Bob. And $\tilde{E}_{I}, \tilde{Y}_{I} ~ (\tilde{E}_{O}, \tilde{Y}_{O})$ are the input (output) ancillary systems of Eve and Yves respectively. These instruments can be described by the sets of Kraus operators 
\begin{subequations}
    \begin{alignat}{4}
        &\{ E_{lm} : \mathcal{H}^{E_I} \otimes \mathcal{H}^{\tilde{E}_I} \to \mathcal{H}^{E_O} \otimes \mathcal{H}^{\tilde{E}_O}\}_{l,m},\\
        &\{ Y_{pq} : \mathcal{H}^{Y_I} \otimes \mathcal{H}^{\tilde{Y}_I} \to \mathcal{H}^{Y_O} \otimes \mathcal{H}^{\tilde{Y}_O}\}_{p,q},
    \end{alignat}
\end{subequations}
respectively such that $ \sum_{m} E_{lm}^{\dagger} E_{lm}\leq \mathbb{I}^{E_I\tilde{E}_I}$, $\sum_{l,m} E_{lm}^{\dagger} E_{lm} = $ $\mathbb{I}^{E_I\tilde{E}_I},$ and similarly for $\{ Y_{pq} \}$. Further, it will be useful rewrite these Kraus operators using the operator-Schmidt decomposition \cite{Operator_decomp,nielsenThesis}:
\begin{subequations}\label{operator_schmidt}
    \begin{alignat}{4}
        E_{lm} &= \sum \limits_{\alpha} E_{lm}^{\alpha} \otimes \tilde{E}_{lm}^{\alpha},\\
        Y_{pq} &= \sum \limits_{\lambda} Y_{pq}^{\lambda} \otimes \tilde{Y}_{pq}^{\lambda},
    \end{alignat}
\end{subequations}
where $\{E_{lm}^{\alpha}:\mathcal{H}^{E_I} \to \mathcal{H}^{E_O}\}_{\alpha}, ~\{\tilde{E}_{lm}^{\alpha}:\mathcal{H}^{\tilde{E}_I} \to \mathcal{H}^{\tilde{E}_O}\}_{\alpha}$ are sets of orthogonal operators, and $Y_{pq}^{\lambda}, \tilde{Y}_{pq}^{\lambda}$ are defined analogously.

Building off of Refs.\,\cite{switch_process_1, Witnessing_causal_nonseparability}, if we input a pure state $\rho_{\psi} = |\psi \rangle\langle \psi|$ to our quantum switch, controlled on the state $\omega = |+\rangle\langle +|$, the causal structure of the setup shown in FIG.\,\ref{fig:YBEA} is described by the process matrix
\begin{equation}\label{full_process_matrix}
    W_{\psi} = W_{\psi}^{ABEYC} \otimes W^{\tilde{E}\tilde{Y}}.
\end{equation}
Here, $ W_{\psi}^{ABEYC} = |w_{\psi}\rangle\langle w_{\psi}|$, defined via
\begin{widetext}
\begin{equation}\label{process_pure_state}
    |w_{\psi}\rangle = \frac{1}{\sqrt2} \big( |\psi \rangle^{Y_I} |\mathbbm1\rangle\!\rangle ^{Y_O B_I}|\mathbbm1\rangle\!\rangle ^{B_O E_I}|\mathbbm1\rangle\!\rangle ^{E_O A_I} |\mathbbm1\rangle\!\rangle ^{A_O C_t} |0\rangle^{C_c} 
    + |\psi \rangle^{A_I} |\mathbbm1\rangle\!\rangle ^{A_O E_I}|\mathbbm1\rangle\!\rangle ^{E_O B_I}|\mathbbm1\rangle\!\rangle ^{B_O Y_I} |\mathbbm1\rangle\!\rangle ^{Y_O C_t} |1\rangle^{C_c} \big),
\end{equation}
\end{widetext}
is the process matrix of the quantum switch containing Alice, Bob Yves and Eve. Further, as mentioned, $W^{\tilde{E}\tilde{Y}}$ is some arbitrary process matrix shared between Eve and Yves detailing all of their possible correlations. Notice the separable nature of $W_{\psi}$. This is because we are assuming Eve and Yves do not, in some sense, have access to the causal structure of the quantum switch: that is, they do not have access to the control qubit. Therefore, their correlations (and instruments) cannot depend on whether $\rho$ is being sent along the $|0\rangle$ path, or the $|1\rangle$ path.

%Intuitively, the CJ isomorphism says that one can think of the temporal evolution of a state through a channel from $\mathcal{L}(\mathcal{H}^{X_I})$ to $\mathcal{L}(\mathcal{H}^{X_O})$ as a spatial teleportation of the state between the same two spaces. Therefore, for our purposes, we can intuitively think of the $W$ as providing the route by which $|\psi\rangle$ is teleported to $\mathcal{H}^{C_t}$.

Let us now write down the CJ operators that describe each lab's operations. For Alice and Bob, being independent, these are given by
\begin{align}
        M^X_k &= (\mathcal{I}\otimes \mathcal{X}_k)\big( |\mathbbm1\rrllangle \mathbbm1|^{X_IX_I'}\big)\\
        &= \sum \limits_{i,j=0}^1  (|i\rangle \langle j|)^{X_I} \otimes (X_k |i\rangle \langle j| X_k^{\dagger})^{X_O} \nonumber\\
        &= \big(|X_k^*\rangle\!\rangle\langle\!\langle X_k^*|^{X}\big)^T, 
\end{align}
where
\begin{equation}
    |X^*_k\rangle\!\rangle^{X} := (\mathbbm1 \otimes X^*_k ) |\mathbbm{1} \rangle\!\rangle^{X}.
\end{equation}
Note that here we are taking $X \equiv X_I X_O\in \{A, B\}$, $X_k \in \{ A_k, B_k \}$ to be the Kraus operators defining the channel $\mathcal{X} \in \{ \mathcal{A}, \mathcal{B} \}$, and we select $\{|i\rangle\} = \{|j\rangle\} = \{|0\rangle, |1\rangle\}$ as our chosen basis for defining $|\mathbbm1\rangle\!\rangle^{X_I X_I'}$.

Let's now move onto the operations making up Eve and Yves's instruments, given in Eq.\,(\ref{EveAndYvesInstruments}). To find the corresponding CJ operators, we choose a basis for their respective Hilbert spaces: 
\begin{subequations}
    \begin{alignat}{4}
        \{|b_{e\tilde{e}}\rangle := |e\rangle^{E_I} |\tilde{e}\rangle^{\tilde{E}_I} \}_{e, \tilde{e}} &\subset \mathcal{H}^{E_I} \otimes \mathcal{H}^{\tilde{E}_I},\\
        \{|c_{y\tilde{y}}\rangle := |y\rangle^{Y_I}  |\tilde{y}\rangle^{\tilde{Y}_I} \}_{y,\tilde{y}} &\subset \mathcal{H}^{Y_I} \otimes \mathcal{H}^{\tilde{Y}_I},
    \end{alignat}
\end{subequations}
where $\{|e\rangle^{E_I} ~|~ e= 0,1 \}, \{|\tilde{e}\rangle^{\tilde{E}_I} \}_{\tilde{e}}$ are our chosen bases for $\mathcal{H}^{E_I}, \mathcal{H}^{\tilde{E}_I}$ respectively, and $\{|y\rangle^{Y_I} ~|~ y = 0,1 \}, \{|\tilde{y}\rangle^{\tilde{Y}_I} \}_{\tilde{y}}$ are our chosen bases for $\mathcal{H}^{Y_I}, \mathcal{H}^{\tilde{Y}_I}$ respectively\footnote{Note that no restrictions are put on the choice of bases for the ancilla systems to preserve generality.}. This allows us to define
\begin{subequations}\label{eavesMaxEntangled}
    \begin{alignat}{4}
        |\mathbbm1\rangle \! \rangle^{E_I\tilde{E}_I, E_I' \tilde{E}_I'} &= \sum \limits_{b_{e\tilde{e}}} |b_{e\tilde{e}} b_{e\tilde{e}}\rangle^{E_I\tilde{E}_I, E_I' \tilde{E}_I'} \nonumber\\
        &= \sum \limits_{e,\tilde{e}} |ee\rangle^{E_I E'_I} |\tilde{e}\tilde{e}\rangle^{\tilde{E}_I \tilde{E}'_I} \nonumber \\
        &= |\mathbbm1\rangle\!\rangle^{E_IE_I'} |\mathbbm1\rangle\!\rangle^{\tilde{E}_I \tilde{E}'_I},\\
        |\mathbbm1\rangle \! \rangle^{Y_I\tilde{Y}_I, Y_I' \tilde{Y}_I'} &= |\mathbbm1\rangle\!\rangle^{Y_IY_I'} |\mathbbm1\rangle\!\rangle^{\tilde{Y}_I \tilde{Y}'_I},
    \end{alignat}
\end{subequations}
where the second expression follows from the same argument of the first. Using the operator-Schmidt decomposition given in Eq.\,(\ref{operator_schmidt}), the corresponding CJ operators are therefore given by
\begin{subequations}
    \begin{alignat}{4}
        M_l^{E\Tilde{E}} &= (\mathcal{I}\otimes \mathcal{E}_l)\Big( |\mathbbm1\rangle\!\rangle \langle \! \langle \mathbbm1|^{E_I\tilde{E}_I, E_I' \tilde{E}_I'} \Big) \nonumber\\
        &= \sum \limits_m \sum_{\alpha,\beta}  \big(|E_{lm}^{\alpha*}\rrllangle E_{lm}^{\beta*}|^E \!\otimes  |\tilde{E}_{lm}^{\alpha*}\rrllangle \tilde{E}_{lm}^{\beta*}|^{\tilde{E}}\big)^T, \\
        M_p^{Y\Tilde{Y}} &= (\mathcal{I}\otimes \mathcal{Y}_p)\Big( |\mathbbm1\rangle\!\rangle \langle \! \langle \mathbbm1|^{Y_I\tilde{Y}_I, Y_I' \tilde{Y}_I'} \Big) \nonumber\\
        &= \sum \limits_q \sum_{\lambda,\mu} \big( |Y_{pq}^{\lambda*}\rrllangle Y_{pq}^{\mu*}|^Y\! \otimes |\tilde{Y}_{pq}^{\lambda*}\rrllangle \tilde{Y}_{pq}^{\mu*}|^{\tilde{Y}}\big)^T,
    \end{alignat}
\end{subequations}
where
\begin{subequations}
    \begin{alignat}{4}
        |E^{\alpha*}_{lm}\rrangle^{E} &= \big(\mathbbm1\otimes E^{\alpha*}_{lm}\big)|\mathbbm1\rrangle^{E_I E_I'},\\
        |\tilde{E}^{\alpha*}_{lm}\rrangle^{\tilde{E}} &= \big(\mathbbm1\otimes \tilde{E}^{\alpha*}_{lm}\big)|\mathbbm1\rrangle^{\tilde{E}_I \tilde{E}_I'},\\
        |Y_{pq}^{\lambda*}\rrangle^Y &= (\mathbbm1\otimes Y_{pq}^{\lambda*})|\mathbbm1\rrangle^{Y_I Y_I'},\\
        |\tilde{Y}_{pq}^{\lambda*}\rrangle^{\tilde{Y}} &= (\mathbbm1\otimes \tilde{Y}_{pq}^{\lambda*})|\mathbbm1\rrangle^{\tilde{Y}_I \tilde{Y}_I'}.
    \end{alignat}
\end{subequations}

We are therefore equipped to see how the initial state of our system $\rho_{\psi}\otimes \omega$ evolves with respect to the process matrix $W_{\psi}$, along with the actions of Alice, Bob, Eve and Yves that we just defined. Following the basis comparison of Alice and Bob,
\begin{widetext}
    \begin{align}
    \rho_{\psi} \otimes \omega &\to 2\sum_{l,p} \sum_{\mathfrak{B} \in \mathsf{C}} \sum_{i,j \in \mathfrak{B}} \Tr_{\bar{C}} \Big[\big(M^A_i \otimes M^B_j \otimes M^E_l \otimes M^Y_p \otimes \mathbb{I}^C\big)^T W_{\psi} \Big] \nonumber \\
    &\,\,= 2\sum_{l,p, m, q} \sum_{\mathfrak{B} \in \mathsf{C}} \sum_{i,j \in \mathfrak{B}} \sum_{\alpha, \beta} \sum_{\lambda,\mu} \Big( \llangle A^*_i| \llangle B^*_j| \llangle E_{lm}^{\alpha*}| \llangle Y_{pq}^{\lambda*}|W^{ABEYC}_{\psi} | A^*_i\rrangle  |B^*_j\rrangle | E_{lm}^{\beta*}\rrangle | Y_{pq}^{\mu*}\rrangle\\
    &\hspace{9cm} \times\llangle \tilde{E}_{lm}^{\alpha*}| \llangle \tilde{Y}_{pq}^{\lambda*}| W^{\tilde{E}\tilde{Y}} | \tilde{E}_{lm}^{\beta*}\rrangle | \tilde{Y}_{pq}^{\mu*}\rrangle\Big),
\end{align}
\end{widetext}
where the factor of two comes from the normalisation of the state following basis comparison and $\Tr_{\bar{C}}$ means we are tracing out all systems but $C = C_t C_c$. This is a rather complicated expression. However, it contains information (about Yves and Eve's correlations) that we don't require for the result of this appendix. To simplify the problem, note first that the positivity of $W^{\tilde{E}\tilde{Y}}$ allows us to write $W^{\tilde{E}\tilde{Y}} = \sum_s |s|^2 |s\rlangle s|$ such that $|s\rangle \in \mathcal{H}^{\tilde{E}} \otimes \mathcal{H}^{\tilde{Y}}$. We can use this to define the Kraus operators $Z_{\mathbf{x}} : \mathcal{H}^{E_I}\otimes \mathcal{H}^{Y_I} \to \mathcal{H}^{E_O} \otimes \mathcal{H}^{Y_O}$ by
\begin{equation}
    Z_{\mathbf{x}} = \sum_{\alpha, \lambda} s \big(\llangle \tilde{E}_{lm}^{\alpha*}| \llangle \tilde{Y}_{pq}^{\lambda*}|\big) |s\rangle E_{lm}^{\alpha} \otimes Y_{pq}^{\lambda},
\end{equation}
where $\mathbf{x} = (l,m,p,q,s) \in \mathbf{X}$, is an index that holds the instrument (measurement) outcome of Eve and Yves and $\mathbf{X}$ is the set containing all such outcomes. In what follows, we forget the structure of these operators and take Eve and Yves's joint action to be some arbitrary set of Kraus operators $\{ Z_{\mathbf{x}}: \mathcal{H}^{E_I}\otimes \mathcal{H}^{Y_I} \to \mathcal{H}^{E_O} \otimes \mathcal{H}^{Y_O} \}$. Doing so means that anything that we find, also holds for the actual situation we have been considering up until this point. 

\subsection{Correlated eavesdroppers cannot learn anything useful using individual attacks}

The conclusion of the previous subsection was that we take to the state $\rho_{\psi} \otimes \omega$ to evolve in the following way during this protocol:
\begin{widetext}
\begin{equation}
    \rho_{\psi} \otimes \omega \to 2\sums \llangle A_i^*|^A \llangle B_j^*|^B\llangle Z_{\mathbf{x}}^*|^{EY} W_{\psi}^{ABEYC} |A_i^*\rrangle^A | B_j^*\rrangle^B |Z_{\mathbf{x}}^*\rrangle^{EY}.
\end{equation}
\end{widetext}
Here,
\begin{align}
    |Z_{\mathbf{x}}^*\rangle\!\rangle &= (\mathbb{I} \otimes Z_{\mathbf{x}}^*) |\mathbbm1\rangle\!\rangle  ^{E_I E_I'}|\mathbbm1\rangle\!\rangle ^{Y_I Y_I'} \\
    &= \sum \limits_{e,y\in\{0,1\}} |ey\rangle^{E_I Y_I} \big(Z_{\mathbf{x}}^* |ey\rangle\big)^{E_O Y_O},
\end{align}
where we followed an analogous argument to Eq.\,(\ref{eavesMaxEntangled}) in order to use $|\mathbbm1\rrangle^{E_IY_I, E_I'Y_I'} = |\mathbbm1\rrangle^{E_IE_I'}|\mathbbm1\rrangle^{Y_IY_I'}$. Now, using Eq.\,(\ref{process_pure_state}) for $W_{\psi}^{ABEYC} = |w_{\psi}\rlangle w_{\psi}|$, we can see explicitly, that after some algebra,
\begin{align}
\begin{aligned}
    \rho_{\psi} \otimes \omega \to \sums|f_{\mathbf{x}ij}\rangle\langle f_{\mathbf{x}ij}|^{C_t C_c},
\end{aligned}
\end{align}
where,
\begin{multline}\label{Fijk}
    |f_{\mathbf{x}ij}\rangle^{C_t C_c} =  \sum \limits_{n\in\{0,1\}} \big[ (\langle n| \otimes \mathbbm1)(B_j\otimes A_i) Z_{\mathbf{x}} |\psi\rangle |n\rangle^{C_t} |0\rangle^{C_c} 
    \\+ (\mathbbm1 \otimes \langle n|) Z_{\mathbf{x}} (B_j \otimes A_i)|n\rangle^{C_t} |\psi\rangle |1\rangle^{C_c}\big].
\end{multline}

We can quickly check our sanity by considering the case when Eve and Yves are not present. That is, when $Z_{\mathbf{x}} \propto \mathbbm1 \otimes \mathbbm1,~\forall \mathbf{x} \in \mathbf{X}$. Here, it turns out that
\begin{equation}
    |f_{\mathbf{x}ij}\rangle^{C_t C_c} \propto \frac{1}{\sqrt{2}} \left( A_i B_j |\psi\rangle^{C_t} |0\rangle^{C_c} + B_j A_i |\psi\rangle^{C_t} |1\rangle^{C_c} \right)
\end{equation}
which is what we'd expect from a quantum switch with two operations \cite{QuantumSwitch_Romero}.

Recall that earlier, we found that when no eavesdroppers are present, measuring the state of the control qubit ${C_c}$ at the end in the $\{|\pm\rangle\}$ basis would always result in $+$. In other words, the probability of measuring $-$, denoted $P(-^{C_c})$, is zero. The question now is, if a correlated Eve and Yves are present, what form must $Z_{\mathbf{x}}$ have in order to ensure $P(-^{C_c}) = 0$? And further, with this form of $Z_{\mathbf{x}}$, can Eve and Yves extract information about the key being shared between Alice and Bob?

\begin{theorem}\label{theorem:Correlated}
For any input state $|\psi\rangle$, $P(-^{C_c}) = 0$ if and only if
\begin{equation}
    Z_{\mathbf{x}} = \sum \limits_{\mu = 0}^3 r_{\mathbf{x}}^{\mu} \sigma_{\mu}\otimes \sigma_{\mu},
\end{equation}
where $r_{\mathbf{x}}^{\mu} \in \mathbb{C}~\forall \mu\in \{0,1,2,3\}, \mathbf{x} \in \mathbf{X}$ and $(\sigma_0, \sigma_1, \sigma_2, \sigma_3) $ $= (\mathbbm{1}, \sigma_x, \sigma_y, \sigma_z)$.
\end{theorem}
\begin{proof}
First, assume that $P(-^{C_c}) = 0$. This means that
\begin{align}\label{proof1}
    \sums \sum \limits_{m\in\{0,1\}} \left| (\langle m|^{C_t} \langle - |^{C_c})|f_{\mathbf{x}ij}\rangle^{C_t C_c}\right|^2 = 0,
\end{align}
where the sum over $m$ comes about because no measurement is being performed on the target qubit $C_t$. This implies that $(\langle m|^{C_t} \langle - |^{C_c})|f_{\mathbf{x}ij}\rangle^{C_t C_c} = 0 ~\forall i,j\in \mathfrak{B}, \mathfrak\in \mathsf{C},m\in\{0,1\}$ and $\forall \mathbf{x} \in \mathbf{X}$. From here on, unless stated otherwise, let $\mathfrak{B}\in\mathsf{C}, \, i,j\in\mathfrak{B},\,\mathbf{x}\in\mathbf{X},\, m\in\{0,1\}$ be arbitrary. Using Eq.\,(\ref{Fijk}), it follows that
\begin{multline}\label{proof2}
    \sum \limits_{n\in\{0,1\}} \big( \langle n|\langle m| (B_j\otimes A_i)Z_{\mathbf{x}}|\psi\rangle |n\rangle 
    \\- \langle m|  \langle n| Z_{\mathbf{x}} (B_j \otimes A_i) |n\rangle |\psi\rangle \big) = 0.
\end{multline}
Suppose we have an arbitrary, pure input state $|\psi\rangle = \alpha |0\rangle + \beta|1\rangle$, where $\alpha,\beta\in\mathbb{C}$ subject to $|\alpha|^2 + |\beta|^2 = 1$. If we show the theorem to be true for this case, it follows that it is true for any mixed state $\rho =\sum_{\psi} p_{\psi} |\psi\rangle \langle \psi|$ by the linearity of the theory. In order to achieve this, we first of all take $|\psi\rangle \neq |0\rangle, |1\rangle$, that is, $\alpha, \beta \neq 0$.

Note that Eq.\,(\ref{proof2}) must hold for both $i,j\in\{0,1\}$ {\it and} $i,j\in\{+,-\}$. Let us first see what we can find out about $Z_{\mathbf{x}}$ when we take $i,j\in\{0,1\}$. In this case, Eq.\,(\ref{proof2}) has the following form:
\begin{multline}\label{General_Condition}
    \delta_{im}\left(\alpha \langle jm| Z_{\mathbf{x}} |0j\rangle + \beta \langle jm| Z_{\mathbf{x}} |1j\rangle\right) 
    \\= (\alpha \delta_{i0} + \beta \delta_{i1}) \langle mj| Z_{\mathbf{x}} |ji\rangle.
\end{multline}
When $i\neq m$, we can quickly see that $\langle 00|Z_{\mathbf{x}}|01\rangle,$ $ \langle 01|Z_{\mathbf{x}}|11\rangle, \langle 10|Z_{\mathbf{x}}|00\rangle, \langle 11|Z_{\mathbf{x}}|10\rangle = 0$. Next, when $m=i$, $Z_{\mathbf{x}}$ is constrained by
\begin{multline}
    \alpha \langle jm| Z_{\mathbf{x}} |0j\rangle + \beta \langle jm| Z_{\mathbf{x}} |1j\rangle \\= (\alpha \delta_{m0} + \beta \delta_{m1}) \langle mj| Z_{\mathbf{x}} |jm\rangle.
\end{multline}
When $(j,m) = (0,0), (1,1)$ we find that $\langle 00|Z_{\mathbf{x}}|10\rangle,$ $  \langle 11|Z_{\mathbf{x}}|01\rangle = 0$ respectively. And, defining $e_{\mathbf{x}} = \langle 10|Z_{\mathbf{x}}|01\rangle, ~d_{\mathbf{x}}= \langle 01|Z_{\mathbf{x}}|10\rangle$, when $(j,m) = (0,1)$, it turns out that $\langle 01|Z_{\mathbf{x}}|00\rangle = \frac{\beta}{\alpha} (e_{\mathbf{x}} - d_{\mathbf{x}})$, and when $(j,m) = (1,0)$, $\langle 10|Z_{\mathbf{x}}|11\rangle = \frac{\alpha}{\beta} (e_{\mathbf{x}} - d_{\mathbf{x}})$. Here we can see why we have not allowed $\alpha =0$ or $\beta = 0$. 

Taking stock so far, $Z_{\mathbf{x}}$ has the form
\begin{equation}\label{Z_Matrix_1}
    Z_{\mathbf{x}} = \kbordermatrix{ & \langle 00| & \langle 01| &\langle 10| & \langle 11|\\
    |00\rangle & a_{\mathbf{x}} & 0 & 0 & b_{\mathbf{x}}\\
    |01\rangle & \frac{\beta}{\alpha}(e_{\mathbf{x}} - d_{\mathbf{x}}) & c_{\mathbf{x}} & d_{\mathbf{x}} & 0\\
    |10\rangle & 0 & e_{\mathbf{x}} & f_{\mathbf{x}} & \frac{\alpha}{\beta}(e_{\mathbf{x}} - d_{\mathbf{x}})\\
    |11\rangle & g_{\mathbf{x}} & 0 & 0 & h_{\mathbf{x}} },
\end{equation}
\begin{comment}
\begin{equation}\label{Z_Matrix_1}
    Z_{\mathbf{x}} = \begin{pmatrix}a_{\mathbf{x}} & 0 & 0 & b_{\mathbf{x}}\\
    \frac{\beta}{\alpha}(e_{\mathbf{x}} - d_{\mathbf{x}}) & c_{\mathbf{x}} & d_{\mathbf{x}} & 0\\
    0 & e_{\mathbf{x}} & f_{\mathbf{x}} & \frac{\alpha}{\beta}(e_{\mathbf{x}} - d_{\mathbf{x}})\\
    g_{\mathbf{x}} & 0 & 0 & h_{\mathbf{x}} \end{pmatrix},
\end{equation}
\end{comment}
where all entries can be complex numbers. This can be simplified further by summing Eq.\,(\ref{proof2}) over $i$ and $j$. This results in
\begin{equation}
    \sum \limits_{n\in\{0,1\}} \langle nm| Z_{\mathbf{x}} |\psi n\rangle = \sum \limits_{n\in\{0,1\}} \langle mn| Z_{\mathbf{x}} |n \psi\rangle,
\end{equation}
which implies
\begin{multline}
    \sum \limits_{n\in\{0,1\}}  \alpha\left(\langle nm| Z_{\mathbf{x}} |0 n\rangle - \langle mn| Z_{\mathbf{x}} |n0\rangle\right) 
    \\= \sum \limits_{n\in\{0,1\}} \beta \left(\langle mn| Z_{\mathbf{x}} |n1\rangle - \langle nm| Z_{\mathbf{x}} |1n\rangle\right),
\end{multline}
which must be true for all $m\in\{0,1\}$. Choosing $m = 0$ and using Eq.\,(\ref{Z_Matrix_1}), we find that $d_{\mathbf{x}} = e_{\mathbf{x}}$. So, we therefore have
\begin{equation}
    Z_{\mathbf{x}} = \begin{pmatrix}a_{\mathbf{x}} & 0 & 0 & b_{\mathbf{x}}\\
    0 & c_{\mathbf{x}} & d_{\mathbf{x}} & 0\\
    0 & d_{\mathbf{x}} & f_{\mathbf{x}} & 0\\
    g_{\mathbf{x}} & 0 & 0 & h_{\mathbf{x}} \end{pmatrix}.
\end{equation}

To finish the derivation, we use the fact that Eq.\,(\ref{proof2}) must also hold for $i,j\in\{+,-\}$. To utilise this, we keep $i,j\in \{0,1\}$ and use the Hadamard operator $H = (\sigma_x + \sigma_z)/\sqrt{2} = H^{\dagger}$ to relate the $x$ and $z$-bases: that is, by replacing $B_j \otimes A_i$ in Eq.\,(\ref{proof2}) with $(H\otimes H)(B_j \otimes A_i)(H\otimes H)$. After some rearranging, we find that
\begin{widetext}
\begin{equation}\label{Z_Constraint_PM}
    \sum \limits_{n\in\{0,1\}} [\delta_{j0} + (-1)^n \delta_{j1}]\langle ji| (H\otimes H)Z_{\mathbf{x}} | \psi n\rangle 
    = \frac12 \frac{(\alpha + \beta)\delta_{i0} + (\alpha - \beta)\delta_{i1}}{\delta_{i0} + (-1)^m \delta_{i1}} \sum \limits_{n\in\{0,1\}} [\delta_{j0} + (-1)^n \delta_{j1}]\langle mn| Z_{\mathbf{x}}(H\otimes H) | ji\rangle.
\end{equation}    
\end{widetext}
Straight away, we can see that the RHS has a dependence on $m$ but the LHS does not. So we can equate the $m = 0$ and $m=1$ cases of the RHS. Doing this, the four cases that come from $i,j\in \{0,1\}$ result in
\begin{align}
    \begin{aligned}
       a_{\mathbf{x}} &= h_{\mathbf{x}},\\
        g_{\mathbf{x}} &= b_{\mathbf{x}} + c_{\mathbf{x}} - f_{\mathbf{x}}.
    \end{aligned}
\end{align}
Updating $Z_{\mathbf{x}}$ and looking at Eq.\,(\ref{Z_Constraint_PM}) when $(i,j,m) = (0,0,0)$ results in $c_{\mathbf{x}} = f_{\mathbf{x}}$ and $b_{\mathbf{x}} = g_{\mathbf{x}}$. Therefore we have
\begin{equation}
    Z_{\mathbf{x}} = \begin{pmatrix}a_{\mathbf{x}} & 0 & 0 & b_{\mathbf{x}}\\
    0 & c_{\mathbf{x}} & d_{\mathbf{x}} & 0\\
    0 & d_{\mathbf{x}} & c_{\mathbf{x}} & 0\\
    b_{\mathbf{x}} & 0 & 0 &a_{\mathbf{x}} \end{pmatrix}
\end{equation}
which can be rewritten as
\begin{multline}
    Z_{\mathbf{x}} = \frac12 \big[ (a_{\mathbf{x}} + c_{\mathbf{x}}) \mathbbm1\otimes \mathbbm1 + (d_{\mathbf{x}} + b_{\mathbf{x}}) \sigma_x \otimes \sigma_x \\
    + (d_{\mathbf{x}} - b_{\mathbf{x}}) \sigma_y \otimes \sigma_y + (a_{\mathbf{x}} - c_{\mathbf{x}}) \sigma_z \otimes \sigma_z \big].
\end{multline}
Further, since the mapping
\begin{equation}
    \begin{cases}
    r_{\mathbf{x}}^0 =a_{\mathbf{x}} + c_{\mathbf{x}},\\
    r_{\mathbf{x}}^1 = d_{\mathbf{x}} + b_{\mathbf{x}},\\
    r_{\mathbf{x}}^2 = d_{\mathbf{x}} - b_{\mathbf{x}},\\
    r_{\mathbf{x}}^3 =a_{\mathbf{x}} - c_{\mathbf{x}}
    \end{cases}
\end{equation}
is invertible and linear, $a_{\mathbf{x}}, b_{\mathbf{x}}, c_{\mathbf{x}}, d_{\mathbf{x}} \in \mathbb{C}$ being independent from one another implies that $r_{\mathbf{x}}^{\mu} \in \mathbb{C}$ are independent from one another. Therefore,
\begin{equation}\label{Z_final}
    Z_{\mathbf{x}} = \sum \limits_{\mu = 0}^3 r_{\mathbf{x}}^{\mu} \sigma_{\mu}\otimes \sigma_{\mu}.
\end{equation}
At this stage, one might notice that we didn't consider all the combinations of $i,j$ in Eq.\,(\ref{Z_Constraint_PM}). It turns out that these give us no further constraints on $Z_{\mathbf{x}}$. To confirm this, we just need to prove the reverse implication of the if and only if statement. If it turns out that we missed some constraints on $Z_{\mathbf{x}}$, $P(-^{C_c})$ would be nonzero in general when using Eq.\,(\ref{Z_final}) for $Z_{\mathbf{x}}$.

So, suppose that $Z_{\mathbf{x}}$ is given by Eq.\,(\ref{Z_final}). Substituting this into $(\langle m|^{C_t} \langle - |^{C_c})|f_{\mathbf{x}ij}\rangle^{C_t C_c}$ and carrying out the sum over $n$ results in
\begin{multline}
    (\langle m|^{C_t} \langle - |^{C_c})|f_{\mathbf{x}ij}\rangle^{C_t C_c} \\= \sum \limits_{\mu = 0}^3 r_{\mathbf{x}}^{\mu}  \langle m | [A_i, \sigma_{\mu} B_j \sigma_\mu] | \psi\rangle = 0
\end{multline}
for all $i,j\in \mathfrak{B}, \mathfrak{B}\in \mathsf{C},m\in\{0,1\}$ and $\forall \mathbf{x}\in\mathbf{X}$ since $[A_i, \sigma_{\mu} B_j \sigma_\mu] = 0 ~ \forall i,j,\mu$.

Finally, for the cases in which $|\psi\rangle \in \{|0\rangle, |1\rangle\}$, notice that what we have shown so far holds for $|\psi\rangle \in \{|+\rangle, |-\rangle\}$. Now, the process matrix defined using Eq.\,(\ref{process_pure_state}) can equivalently be formulated in the $x$-basis, and Alice and Bob's channels are invariant under this basis change. So, after converting everything to the $x$-basis, what was the situation in which $|\psi\rangle = |\!\!+\!\!/-\rangle$ becomes that of when $|\psi\rangle = |0/1\rangle$. Thus, since $Z_{\mathbf{x}}$ [Eq.\,(\ref{Z_final})] has the same form when it is changed from the $z$-basis to the $x$-basis, the result also holds for this case.
\end{proof}

It is difficult to have any intuition about what Eve and Yves's measurement would look like physically. This is because, in the previous subsection, in anticipation of simplifying this proof, we ignored the correlations shared by Eve and Yves through their ancillary process matrix $W^{\tilde{E}\tilde{Y}}$. A little intuition can be gained, however, by considering what happens when only one of the coefficients $r_{\mathbf{x}}^{\mu}$ is nonzero for each $\mathbf{x}$.  In this scenario, if one of Eve or Yves performs $\sigma_{\mu}$, then the other eavesdropper must do the same if they want to go undetected. 

As discussed earlier, an ancillary process matrix shared between Eve and Yves is likely necessary to understand what is happening here physically with regards to the correlations between the eavesdroppers' measurements. Either way, and to reiterate, since the approach taken here considers no physical constraints on the correlations between Eve and Yves, operations with physical correlations should exist as a subset of the ones we have derived here.

\begin{figure*}
    \centering
    \includegraphics[width=0.8\textwidth]{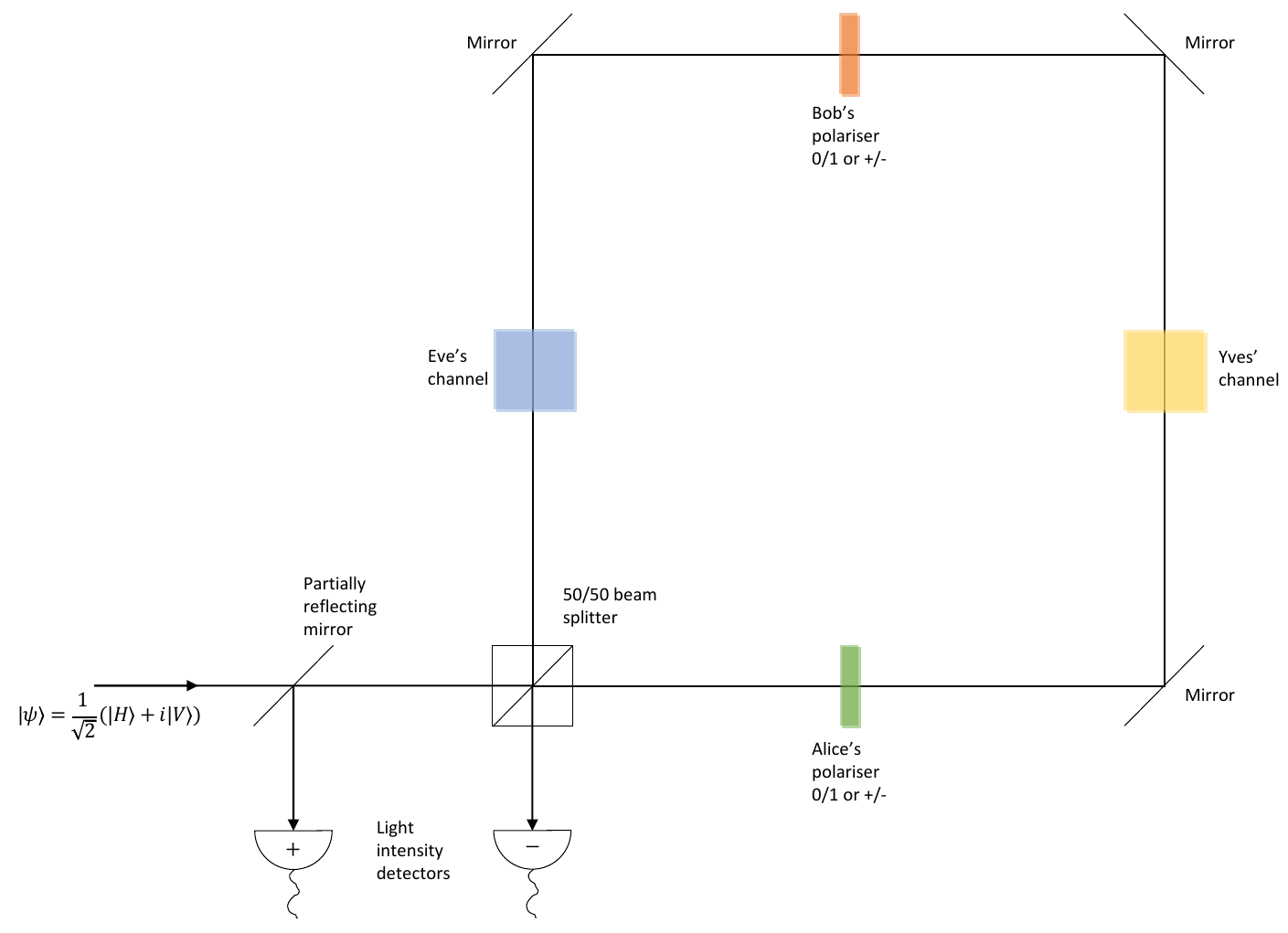}
    \caption{The results derived in this work can be simulated using polarised light to share a key between Alice and Bob, and a Sagnac interferometer to induce the indefinite causal order. Within the interferometer, polarising filters are orientated to correspond to all the valid measurement outcomes Alice and Bob obtain during the protocol.}
    \label{fig:Sagnac}
\end{figure*}

So, can Eve and Yves gain any information about Alice and Bob's shared key using Eq.\,(\ref{Z_final})? To answer this, we calculate $P(Z_{\mathbf{x}}, A_i, B_j) = \langle f_{\mathbf{x}ij}|f_{\mathbf{x}ij}\rangle$ for $i,j\in\{0,1\}$ and $i,j\in\{+,-\}$. These are given by:
\begin{widetext}
\begin{equation}
    P(Z_{\mathbf{x}},A_i,B_j) = \begin{cases} \frac{|\langle i|\psi\rangle|^2}{2} \!\left( |r_{\mathbf{x}}^0 + r_{\mathbf{x}}^3|^2 \delta_{ij} + |r_{\mathbf{x}}^1 + r_{\mathbf{x}}^2|^2 \delta_{i\bar{j}} \right)\!, ~~ i,j\in \{0,1\}\\
    \frac{|\langle i|\psi\rangle|^2}{2} \!\left( |r_{\mathbf{x}}^0 + r_{\mathbf{x}}^1|^2 \delta_{ij} + |r_{\mathbf{x}}^2 + r_{\mathbf{x}}^3|^2 \delta_{i\bar{j}}\right)\!, ~~ i,j\in \{+,-\},\end{cases}
\end{equation}
\end{widetext}
where $\bar{j}$ denotes ``not $j$". At first glance, it appears that Eve and Yves have access to some information about Alice and Bob's key. However, note first that distinguishing between the two cases of $i,j\in \{0,1\}$ and $i,j\in \{+,-\}$ is of no use as Alice and Bob publicly discuss which basis they measured in after they have done so. Secondly, although Eve and Yves could alter $r_{\mathbf{x}}^{\mu}$ however they like, the only additional information they could gain is about whether each bit of Alice and Bob's key agree or not. Therefore, they can still do no better than a guess to determine the key. This intuitive argument is backed up by a calculation of the mutual information between the eavesdroppers and either Alice or Bob. In both cases, this turns out to be zero.

Finally, it should be noted this protocol has a weakness if we allow the eavesdroppers alter the causal structure set up by Alice. That is, if they had access to the control qubit that dictates the indefinite causal order of the quantum switch. In this case, Eve and Yves could perform a similar attack to the one described in Sec.\,\ref{sec:DefiniteCausalOrder?}. That is, if the eavesdroppers return Alice's qubit unaffected, and in an indefinite causal order, whilst, independently sending out a probe state to Bob, they can learn about Bob's key bit, without inducing any $``-"$ measurement results in the control mode. Having said this, it seems as though this eavesdropping strategy can be detected by monitoring $\omega$ in the cases when Alice and Bob disagree on their basis choice. The analysis of such attacks, in which eavesdroppers can infiltrate the causal structure of the protocol, is beyond the scope of this paper.

\subsection{Error rates}\label{sec:ErrorRateAppendix}

In this appendix, assuming our input target state is\footnote{Where $|y\rangle = (|0\rangle + i |1\rangle)/\sqrt2$.} $\rho = |y\rlangle y|$, we prove that  if only Eve is present and induces a probability $P_{\text{detect}} = P(\omega=|-\rlangle-|)$ of being detected, then the error rate of the sifted key shared by Alice and Bob is
\begin{equation}
    P_{\text{error}} = 2P_{\text{detect}}.
\end{equation}

To show this, we use Eq.\,(\ref{keepState}) and let Eve's Kraus operators be written as
\begin{equation}
    E_j = e^{(j)}_{bb}|b\rlangle b| + e^{(j)}_{b\bar{b}}|b\rlangle \bar{b}| + e^{(j)}_{\bar{b}b}|\bar{b}\rlangle b| + e^{(j)}_{\bar{b}\bar{b}}|\bar{b}\rlangle \bar{b}|,
\end{equation}
where $e^{(i)}_{bb'} \in \mathbb{C},~ \forall i$ and $b,b' \in \mathfrak{B}$. Note that we are taking $e^{(i)}_{bb'}$ to depend on the specific basis: if $b,b'\in\{0,1\}$, $e^{(i)}_{bb'}$ is generically not the same as when $b,b' \in\{+,-\}$. As always, we take $\bar{b}$ to mean ``not $b$". From Eq.\,(\ref{keepState}), we see that the probability of detection is given by
\begin{equation}
    P_{\text{detect}} = \frac12\sum \limits_j \sum \limits_{\mathfrak{B} \in \mathsf{C}} \sum \limits_{i,k\in \mathfrak{B}} \langle y|[A_{i}, E_{j}, B_{k}]^{\dagger} [A_{i}, E_{j}, B_{k}]|y\rangle.
\end{equation}
After some algebra, this can be shown to be equal to
\begin{equation}
    P_{\text{detect}} = \frac18 \sum_j \Big( \big|e^{(j)}_{01}\big|^2 + \big|e^{(j)}_{10}\big|^2 + \big|e^{(j)}_{+-}\big|^2 + \big|e^{(j)}_{-+}\big|^2\Big).
\end{equation}

Once again using $\rho_{\text{sifted}}$ from Eq.\,(\ref{keepState}), we can calculate the probability of error between Alice and Bob's sifted key bit. This corresponds to Alice and Bob disagreeing on their measurement outcomes. So, taking $\rho_{\text{sifted}}|_{a\neq b}$ to denote the terms in $\rho_{\text{sifted}}$ where Alice and Bob's outcomes disagree, it can be shown that
\begin{align}
    P_{\text{error}} &= \Tr \big( \rho_{\text{sifted}}|_{a\neq b}\big) \nonumber\\
    &= \frac14 \sum_j \Big( \big|e^{(j)}_{01}\big|^2 + \big|e^{(j)}_{10}\big|^2 + \big|e^{(j)}_{+-}\big|^2 + \big|e^{(j)}_{-+}\big|^2\Big) \nonumber\\
    &= 2P_{\text{detect}}.
\end{align}

%%%%%%%%%%%%%%%%%%%%%%%%%
%%%%%%%%%%%%%%%%%%%%%%%%%

\section{Experimental simulation}\label{sec:Exp}

Figure \ref{fig:Sagnac} shows a possible experimental setup to simulate some of the results derived. The idea is to use photon polarisation (in the horizontal-vertical basis, with $|H\rangle =: |0\rangle, |V\rangle =: |1\rangle$) as the target state $\rho$, initially in the state $|\psi\rangle\langle\psi|$, that is acted on by Alice, Bob, Eve and Yves. As is mentioned in the main text, if we wanted Alice and Bob to have approximately equal numbers of 0s and 1s, we can take our input state to be $\mathbbm{1}/2$. This can be achieved by taking it to be $|i\rangle$ half of the time and $|-i\rangle$ the remainder of the time. These correspond to left and right circularly polarised light respectively: $|\psi\rangle \in \{|\pm i\rangle =  (|H\rangle \pm i|V\rangle)/\sqrt{2}\}$. The control state $\omega$ is taken to be the path degree of freedom induced by a beamsplitter. Using a 50/50 beamsplitter corresponds to taking $\omega = |+\rangle\langle +|$ with $|0\rangle$ corresponding to reflection and $|1\rangle$ to transmission.

Recall that Alice and Bob perform projective measurements in either the $x$ or $z$-basis. This is difficult to do non-destructively and even more difficult to do while keeping the photon continuing around the Sagnac interferometer in its original superposition of paths. That being said, there has been recent progress along these lines \cite{InstrumentsInICO}. In this appendix, however, we just consider a simulation of projective measurements using polarisers. This means we can simulate the statistics that the measurements of Alice, Bob, Eve and Yves would produce.

Explicitly, when Alice and Bob measure in the $z$-basis, we use polarisers orientated at $0$ and $\pi/2$ which correspond to measurement outcomes of 0 and 1 respectively. Likewise, when measuring in the $x$-basis, polarisers being orientated at $\pm \pi/4$ correspond to measurement outcomes of $\pm$ respectively. The probability of Alice and Bob measuring $i,j$ can be taken to be the ratio of the total intensity $I_{\text{exit}}(i,j)$ of light exiting the interferometer to that of it entering $I_{\text{enter}}$:
\begin{equation}
    P(A_i, B_j) = \frac{I_{\text{exit}}(i,j)}{I_{\text{enter}}}.
\end{equation}
Here, the dependence of $I_{\text{exit}}(i,j)$ on $i,j$ highlights that the interferometer is setup with Alice and Bob's polarisers being orientated correspondingly to the measurement outcomes $i,j$ respectively. Since Alice and Bob only keep measurement results when they have publicly confirmed that they measured in the same basis, there are eight permutations when ignoring Eve and Yves. These are given in the Table\,\ref{table}.
\begin{table}[ht!]
\caption{Table detailing the eight polariser orientations that correspond to the possible measurement outcomes that Alice and Bob can obtain when they measure in the same basis.}
\label{table}
\vspace{3mm}
\centering
\begin{tabular}{c|c}
Alice polariser orientation & Bob polariser orientation \\ \hline
0                           & 0                         \\
0                           & $\frac{\pi}{2}$           \\
$\frac{\pi}{2}$             & 0                         \\
$\frac{\pi}{2}$             & $\frac{\pi}{2}$           \\
$\frac{\pi}{4}$             & $\frac{\pi}{4}$           \\
$\frac{\pi}{4}$             & $-\frac{\pi}{4}$          \\
$-\frac{\pi}{4}$            & $\frac{\pi}{4}$           \\
$-\frac{\pi}{4}$            & $-\frac{\pi}{4}$         
\end{tabular}
\end{table}

The key feature of this protocol involves the measurement of the control state in the $\pm$ basis. Noticing that, after going through the main part of the Sagnac interferometer, the path that the light exits the 50/50 beamsplitter along, is controlled by the path qubit in the $x$-basis. That is, the $|+\rangle$ component is transmitted through the beamsplitter, whereas the $|-\rangle$ component is reflected. Therefore, placing a detector in the reflected arm corresponds to the $-$ outcome and, after a partially reflecting mirror, a detector in the transmitted arm corresponds to the $+$ outcome. The probability of measuring the eavesdroppers comes from the probability of measuring the control qubit state to be $-$. Therefore, for each run of the experiment (each permutation of polariser angles), the ratio of the intensity in the $-$ arm to the total intensity exiting the interferometer is what is required. As a sanity check, this should always be zero when Eve and Yves are not present. As mentioned before, in order to exploit the features of indefinite causal order, the coherence length of the light used should be significantly longer than the path length of the interferometer. A laser can be used to achieve this.

%%%%%%%%%%%%%%%%%%%%%%%%%%%%%%%%%

%BIBLIOGRAPHY

\bibliography{apssamp}
\end{document}